\documentclass{siamart1116}

\usepackage{lipsum}
\usepackage{amsfonts}
\usepackage{graphicx}
\usepackage{epstopdf}
\usepackage{algorithmicx}
\usepackage{algpseudocode}
\usepackage{graphicx,epstopdf} % <- Preamble
\usepackage[caption=false]{subfig} % <- Preamble

\algdef{SE}[DOWHILE]{Do}{doWhile}{\algorithmicdo}[1]{\algorithmicwhile\ #1}%
\ifpdf
  \DeclareGraphicsExtensions{.eps,.pdf,.png,.jpg}
\else
  \DeclareGraphicsExtensions{.eps}
\fi

\usepackage{amsmath}
\usepackage{amssymb}
\usepackage{subfig}
\usepackage{enumitem}
\usepackage{mathtools}
\usepackage{multirow}

\usepackage{wrapfig}

\numberwithin{theorem}{section}

\newsiamremark{remark}{Remark}

\newcommand{\TheTitle}{Triangulated Surface Denoising using High Order Regularization with Dynamic Weights}
\newcommand{\TheAuthors}{Zheng Liu, Rongjie Lai, Huayan Zhang and Chunlin Wu}

\headers{Triangulated Surface Denoising using High Order Regularization}{\TheAuthors}

\title{{Triangulated Surface Denoising using High Order Regularization with Dynamic Weights}\thanks{Submitted to the editors DATE.}}
\author{Zheng Liu%
  \thanks{National Engineering Research Center of Geographic Information System, China University of Geosciences, Wuhan, China (\email{liu.zheng.jojo@gmail.com}).}%
  \and
  Rongjie Lai%
  \thanks{Department of Mathematics, Rensselaer Polytechnic Institute, United States
    (\email{lair@rpi.edu}).}
  \and
  Huayan Zhang%
    \thanks{School of Computer Science and Software, Tianjin Polytechnic University, Tianjin, China
    (\email{zhy101@mail.ustc.edu.cn}).}
    \and
  Chunlin Wu%
    \thanks{The Corresponding Author, School of Mathematical Sciences, Nankai University, Tianjin, China
    (\email{wucl@nankai.edu.cn}).}
}

\usepackage{amsopn}

\ifpdf
\hypersetup{
  pdftitle={\TheTitle},
  pdfauthor={\TheAuthors}
}
\fi

\begin{document}

\maketitle

% REQUIRED
\begin{abstract}
    Recovering high quality surfaces from noisy triangulated surfaces is a fundamental important problem in geometry processing.
    Sharp features including edges and corners can not be well preserved in most existing denoising methods except the recent total variation (TV) and $\ell_0$ regularization methods.
    However, these two methods have suffered producing staircase artifacts in smooth regions.
    In this paper, we first introduce a second order regularization method for restoring a surface normal vector field, and then propose a new vertex updating scheme to recover the desired surface according to the restored surface normal field.
    The proposed model can preserve sharp features and simultaneously suppress the staircase effects in smooth regions which overcomes the drawback of the first order models.
    In addition, the new vertex updating scheme can prevent ambiguities introduced in existing vertex updating methods.
    Numerically, the proposed high order model is solved by the augmented Lagrangian method with a dynamic weighting strategy.
    Intensive numerical experiments on a variety of surfaces demonstrate the superiority of our method by visually and quantitatively.
\end{abstract}

% REQUIRED

% REQUIRED

\section{Introduction} \label{sec:introduction}
Triangulated surfaces are used in a variety of fields, such as computer graphics \cite{Pmp2010}, computer-aided design \cite{Barnhill1985Surfaces}, computer vision \cite{Bronstein2008} and many others \cite{kimmel1998computing,kaus2004automated,lai2011framework}.
Triangulated surfaces are usually generated by some digital scanner devices or triangulation algorithms \cite{Lorensen1987}.
However, even with high-fidelity scanners, the scanning process inevitably produces noise due to local measurement errors \cite{Hope92}.
Such noise affects the quality of surfaces and usually cause errors in downstream geometry applications, such as surface reconstruction, segmentation and visualization \cite{Wang2012A}.
Thus, how to effectively remove noise to recover high quality surfaces is one of the most fundamental tasks in geometry processing.
In practice, it is difficult to distinguish noise and sharp features as they are of high frequency information. Meanwhile, it is also important to preserve smooth regions such as quadratic patches.
Therefore, it is still quite challenging to remove noise while preserving sharp features and smoothly curved regions.

Filtering schemes, which can be roughly classified into two categories (isotropic and anisotropic methods), are widely applied in surface denoising.
The isotropic methods \cite{Field1988Laplacian,taubin1995signal,desbrun1999implicit,nehab2005efficiently} are classical and simple, among which Laplacian smoothing \cite{Field1988Laplacian} is typical.
Laplacian smoothing is the process of reducing the surface area.
It smoothes the surface to remove the noise without considering surface geometric features.
Thus, it, as well as other isotropic methods, suffers surface shrinkage and blurs geometric features.
Later on, a variety of anisotropic methods \cite{clarenz2000anisotropic,desbrun2000anisotropic,bajaj2003anisotropic,fleishman2003bilateral,jones2003non,Zheng:11,solomon2014general} were proposed to provide geometric features preservation.
Compared to the isotropic methods, the anisotropic methods are more effective for preserving geometric features.
However, when the noise level increases, the anisotropic methods usually fail to produce satisfactory results.
Especially, this drawback is more severe for surfaces containing sharp features.

Variational methods are another kind of techniques for triangulated surface denoising proposed recently. To keep the sharp features, the variational models use sparsity regularization term.
Inspired by the great success of total variation (TV) regularization in image processing \cite{Rudin1992Nonlinear}, several researchers extended it to triangluated surface denoising.
The authors of \cite{elsey2009analogue} presented an analogue of TV by minimizing the absolute value of Gauss curvature.
Very recently, Zhang et al. proposed in \cite{Zhang:15} a vectorial TV based model on face normal field over triangulated surfaces.
This method achieved impressive results for preserving sharp features.
Another sparsity regularization is $\ell_0$ quasi-norm. Indeed, He and Schaefer \cite{He13} extended $\ell_0$ minimization \cite{Xu2011Image} to triangulated surfaces for preserving sharp features.
These methods achieve impressive results for surfaces consisting of flat regions and sharp features, e.g., polyhedron surfaces.
However, if a surface has smoothly curved regions, they tend to flatten the smooth regions. The reason is the staircase effect of the sparsity regularization in the gradient field.
The staircase effect of TV in image processing has been studied both from a theoretical and experimental points in previous works; see \cite{Chambolle1997Image,Chan2000High,Lysaker2003Noise,Lysaker2006Iterative,Maso2007A} and references therein.
To overcome this disadvantage of TV, high order PDEs \cite{Scherzer1998Denoising,Lysaker2003Noise,Hinterberger2006Variational,Bergounioux2010A,Lai13} and combination methods of TV and high order models
\cite{Chambolle1997Image,You2000Fourth,Chan2000High,Lysaker2003Noise,Lysaker2006Iterative,Hinterberger2006Variational,Chan2010A,Papafitsoros2014A} have been used in image processing community.
However, to our best knowledge, very few of high order models or combinations are known over triangulated surfaces.

Wavelet frame methods have been successfully applied in image restoration \cite{Cai2012Image,Cai2015Image}.
Recently, Dong et al. \cite{dong2009wavelet,Dong2016Multiscale,dong2017sparse} extended the wavelet frame methods to triangulated surfaces.
Their tight wavelet frame systems are potentially effective in many geometry applications, such as denoising and semi-supervised clustering.
Especially, for surface denoising, Dong et al. \cite{Dong2016Multiscale} proposed multiscale representation of surfaces using wavelet frames, which can achieve impressive denoising results for piecewise smooth surfaces with multiscale details.
Yang and Wang \cite{Yang2017A} proposed a wavelet frame based variational model in \cite{Yang2017A}.
Their method can effectively remove mixed Gaussian and impulse noise for the $\ell_1 + \ell_2$ fidelity term of their model.
However, the existing wavelet frame based methods have difficulty to recover surfaces consisting of sharp features.

Among the methods mentioned above, there are some methods belonging to two-stage methods, i.e., face normal filtering followed by updating vertices \cite{Taubin2001Linear,Yagou2002Mesh,Shen2004Fuzzy,Lee2006Feature,Sun:07,Sun:08,Zheng:11,Zhang:15,Zhang2015Guided}.
The difference between these two-stage methods is in their normal filtering strategies, e.g., a mean and median normal filter was applied in \cite{Yagou2002Mesh}, \cite{Sun:07} adopted trimmed quadratic weights for averaging the normals, Zhang et al. in \cite{Zhang:15} used a TV based model to filter a face normal field.
All the normal filtering strategies can either deal with smooth regions or sharp features well.
Moreover, all these two-stage methods use almost the same vertex updating model, which originated from Taubin \cite{Taubin2001Linear}, and has a beautiful implementation by Sun et al. \cite{Sun:07}.
When the noise level is low, the approach by Sun et al. \cite{Sun:07} can achieve good results.
However, when the noise level increases, the recovered vertex positions deviate far from those of the clean surface.
In this situation, method of Sun et al. \cite{Sun:07} suffers producing frequent foldovers.
Moreover, large scale noise in random directions make the matter even worse.
This is due to the method \cite{Sun:07} neglects the orientations of triangle face normals, which leads vertex updating ambiguity problem; see \cref{sec:VertexUpdateModel} for the explanation of this ambiguity.

As we can see, the aforementioned surface denoising schemes including the filtering, variational and wavelet frame methods can either properly handle smooth regions or sharp features. However, it is still quite challenging to handle both smooth regions and sharp features well.
In this paper, we propose a high order regularization model by introducing a new second order difference operator over triangulated surfaces.
The proposed model with a well-designed weighting function is applied to the surface face normal field, which has crucial advantage in handling surfaces consisting of both smooth regions and sharp features.
It preserves sharp features well and substantially suppresses the staircase effect.
It is numerically solved by the operator splitting and augmented Lagrangian method.
The weighting function enhances the sparsity of the proposed high order model and is implemented by a dynamic weights strategy.
After restoring the face normals, the surface vertices should be updated to match the filtered face normals.
Last but not least, a new vertex updating method is presented.
Compared to the traditional vertex updating method \cite{Sun:07}, our new method can eliminate ambiguities and reconstruct much better triangulated surfaces.
To summarize, the contributions of the paper are listed as follows:
\begin{itemize}
  \item We introduce a new second order difference operator and its adjoint operator in piecewise constant function space over triangulated surfaces.
        To the best of our knowledge, this second order operator is firstly defined over triangulated surfaces.
  \item We introduce a novel normal filtering model using the second order regularization with a well-designed weighting function, which can preserve sharp features and simultaneously prevent the staircase effect in smooth regions.
  \item We propose a new vertex updating method to recover surface vertices.
        The proposed method significantly reduces foldovers compared to the existing vertex updating methods.
\end{itemize}

The rest of this paper is organized as follows.
In \cref{sec:TVandHOTV}, we briefly review TV based models in image processing and reweighted $\ell_1$ minimization.
\Cref{sec:pieceWiseConstantFunctionSpace} provides the definitions of a new second order difference operator and two high order regularization models in piecewise constant function spaces.
The differences of this second order operator and the Laplace operator are discussed at the end of \cref{sec:pieceWiseConstantFunctionSpace}.
In \cref{sec:normalFiltering}, we present a high order regularization normal filtering model with a well-designed weighting function.
An augmented Lagrangian method is applied to solve the variational model with a dynamic weights strategy.
In \cref{sec:VertexUpdateModel}, a new vertex updating method is introduced for recovering the vertex positions with respect to the filtered face normals.
Our two-stage denoising method is discussed and compared to typical existing methods both qualitatively and quantitatively in \cref{sec:results}.
\Cref{sec:conclusion} concludes the paper.

\section{TV Based Models and Reweighted $\ell_1$ Minimization} \label{sec:TVandHOTV}
In this section, we present TV based models and reweighted $\ell_1$ minimization, since they are closely related to our approach.

\subsection{TV, vectorial TV and high order models for images}
Since the pioneering work of Rudin et al. \cite{Rudin1992Nonlinear}, TV has been proven very successful in image processing for its excellent edge-preserving property
\cite{Rudin1992Nonlinear,Lysaker2003Noise,Lysaker2006Iterative,Lai13}.
The TV denoising model (ROF) aims at solving
\begin{equation} \label{c-tvModel}
 \min \limits_{u} \int_{\mathrm{\Omega}} |\nabla u| +\frac{\alpha}{2} \int_{\mathrm{\Omega}} (u-f)^2,
\end{equation}
where $f$ is an observed noisy image, $\int_{\mathrm{\Omega}} |\nabla u|$ is the TV regularization and $\alpha$ is a positive fidelity parameter.
For $\mathfrak{M}$-channel images $\mathbf{u}, \mathbf{f} : \mathrm{\Omega} \rightarrow \mathbb{R}^{\mathfrak{M}}$, where $\mathbf{u}=(u_{1}, u_{2}, \ldots, u_{\mathfrak{M}})$ and $\mathbf{f}=(f_{1}, f_{2}, \ldots, f_{\mathfrak{M}})$, the model \cref{c-tvModel} can be naturally extended to its vectorial version for color image processing as follows:
\begin{equation} \label{c-vtvModel}
 \min \limits_{\mathbf{u}} \int_{\mathrm{\Omega}} \Big(\sum\limits_{i=1}^{\mathfrak{M}}|\nabla u_{i}|^{2}\Big)^{\frac{1}{2}} + \frac{\alpha}{2} \int_{\mathrm{\Omega}} |\mathbf{u}-\mathbf{f}|^2.
\end{equation}
The regularization of model \cref{c-vtvModel} referred as vectorial TV has been discussed in \cite{Sapiro:1996,Blomgren1998Color,chan2001total,Bresson2008}.
Both the objective functionals are coercive, proper, continuous, and strictly convex.
Thus, the problems \cref{c-tvModel} and \cref{c-vtvModel} have respectively, a unique minimizer.

A well known drawback of the above TV and vectorial TV models is the staircase effect \cite{Chambolle1997Image,Chan2000High,Lysaker2003Noise,Lysaker2006Iterative}.
To overcome this, high order models such as Lysaker-Lundervold-Tai (LLT) model \cite{Lysaker2003Noise} and Total Generalized Variation (TGV) model \cite{Bredies2010Total}, have been studied \cite{Scherzer1998Denoising,Lysaker2003Noise,Hinterberger2006Variational,Bergounioux2010A,Bredies2010Total}.
The idea is essentially to introduce high order derivatives to the energy regularization.
High order models in general perform well in recovering smooth regions, but they cannot compete with TV in dealing with discontinuous edges. A natural solution is to combine TV and high order models \cite{Chambolle1997Image,You2000Fourth,Chan2000High,Lysaker2003Noise,Lysaker2006Iterative,Hinterberger2006Variational,Chan2010A,Papafitsoros2014A}.
%The balance is usually implemented by a weighting parameter or function, which needs to be tuned carefully.
For examples, Lysaker and Tai \cite{Lysaker2006Iterative} used a convex combination of TV with LLT \cite{Lysaker2003Noise}.
In \cite{Chan2000High}, Chan et al. presented a model combining a TV term with a weighted Laplacian term to reduce the staircase effect while preserving sharp edges.
A model using infimal-convolution of the TV and high order term, was proposed by Chambolle and Lions in \cite{Chambolle1997Image}, in which the TV term was used to keep sharp edges while the high order term preserves smooth regions.
The key of these methods is to balance the contribution of the TV and high order term.
The balance is usually implemented by a weighting parameter or function, which needs to be tuned carefully.

\subsection{Reweighted $\ell_1$ Minimization}
The reweighted $\ell_1$ minimization was first presented by Cand\`{e}s et al. in \cite{Cand2008} to enhance the sparsity in sparse signal recovery. It outperforms $\ell_1$ minimization in situations where substantially fewer measurements are used to recover a signal.

The key of the reweighted $\ell_1$ minimization is to solve a sequence of weighted $\ell_1$ minimization problems
\begin{equation} \label{eqn:reweightedL1}
   x^{(k)} = \mathrm{arg} \min\limits_{x\in \mathbb{R}^{n}} \| W^{(k)} x \|_{1} \quad \textmd{s.t.} \quad Ax=b,
\end{equation}
where $W^{(k)}=\mathrm{diag}(w^{k}_1,...,w^{k}_n)$ is updated according to $x^{(k-1)}$.
%Denoting $w^{(k)} = (w^{k}_{1},...,w^{k}_{n})^{'}$ as the vector of weights,
%the solution of \eqref{eqn:reweightedL1} is set to be $x^{(k+1)}$.
%Then the weights $w^{(k+1)}$ are updated according to $x^{(k+1)}$.
Although there are a variety of reweighted $\ell_1$ algorithms proposed to update the weights \cite{Cand2008,wang2010sparse,Zhao2012Reweighted}, as a rough rule of thumb,
the weights should be inversely proportional to signal magnitudes \cite{Cand2008}.
%the weights should be related to signal magnitudes on the $\ell_1$ penalty function inversely \cite{Cand2008}. %\rcomment{Perhaps write the formula or explain the way of updating $W$}
For example, the reweighted method proposed by Cand\`{e}s et al. in \cite{Cand2008} is as follows :
%For example, Caned\`{e}s et.al proposed one reweighted method as
\begin{equation*}
  w^{k}_{i} = \frac{1}{|x^{k-1}_{i}|+\epsilon}, \  i=1,... ,n, \  \textmd{for} \, \epsilon>0.
\end{equation*}

\section{Discrete High Order Regularization Models in Piecewise Constant Function Spaces Over Triangulated Surfaces} \label{sec:pieceWiseConstantFunctionSpace}
In this section, we introduce some notations followed by definitions of piecewise constant function spaces and difference operators over triangulated surfaces.
The discrete high order models in piecewise constant spaces are presented and discussed.

\subsection{Notations}
Let $M$ be a compact triangulated surface of arbitrary topology with no degenerate triangles in $\mathbb{R}^3$.
The set of vertices, edges and triangles of $M$ are denoted as
$\{v_i:i=0,1,\cdots,\mathrm{V}-1\}$,
$\{e_i:i=0,1,\cdots,\mathrm{E}-1\}$ and
$\{\tau_i:i=0,1,\cdots,\mathrm{T}-1\}$, respectively.
Here $\mathrm{V}$, $\mathrm{E}$ and $\mathrm{T}$ are the numbers of
vertices, edges and triangles of $M$, respectively.
If $v$ is an endpoint of an edge $e$, then we write it as $v\prec e$.
Similarly, $e\prec \tau$ denotes that $e$ is an edge of a triangle
$\tau$; $v\prec \tau$ denotes that $v$ is a vertex of a triangle
$\tau$.

Denote the $1$-ring of the triangle $\tau_i$ as $D_1(\tau_i)$, which is the set of the triangles sharing some common edges with $\tau_i$ indicated as green triangles in \cref{fig:neighborhood}(a).
Let $B_{1}(\tau_i) = \{l_{j,\tau_i}:i=0,1,\cdots,\mathrm{T}-1;j=0,1,2\}$ be the set of lines connecting the barycenter and vertices of $\tau_i$, where $j$ counterclockwise marks the vertex contained in $\tau_i$.
Namely, $l_{j,\tau_i}$ is the line connecting the barycenter of $\tau_i$ and the vertex marked as $j$ in $\tau_i$.
Let $B_2(\tau_i)$ be the set of lines connecting vertices of $\tau_i$ and barycenters of triangles in $D_1(\tau_i)$ indicated as blue lines in \cref{fig:neighborhood}(b).
Write the 1-disk of the vertex $v_i$ as $M_1(v_i)$ denoting
the indices of triangles containing $v_i$ indicated as gray triangles in \cref{fig:neighborhood}(c).
We write the 1-neighborhood of vertex $v_i$ as $N_{1}(v_i)$, which is the set of vertices connecting to $v_i$ indicated as orange vertices in \cref{fig:neighborhood}(d).

\begin{figure}[tbhp]
  \centering
  \subfloat[]{\label{fig:neighborhood-a}\includegraphics[width=0.25\textwidth]{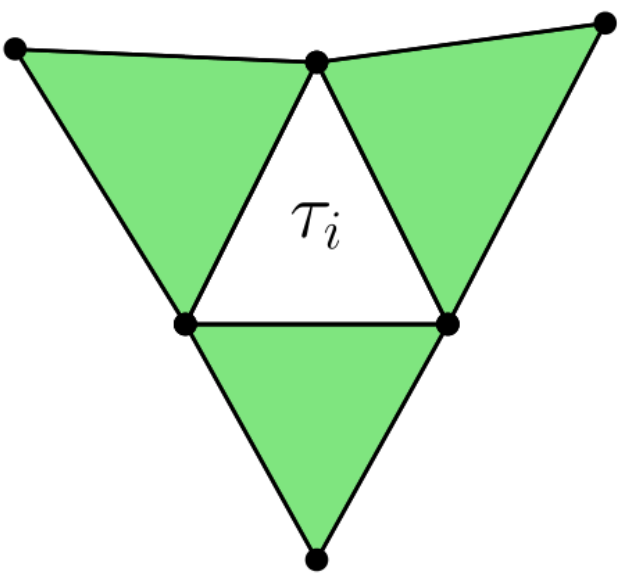}}
  \subfloat[]{\label{fig:neighborhood-b}\includegraphics[width=0.25\textwidth]{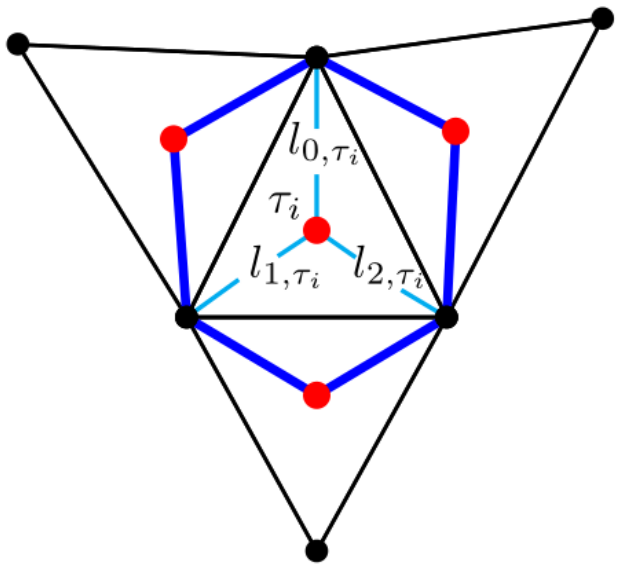}}
  \subfloat[]{\label{fig:neighborhood-c}\includegraphics[width=0.25\textwidth]{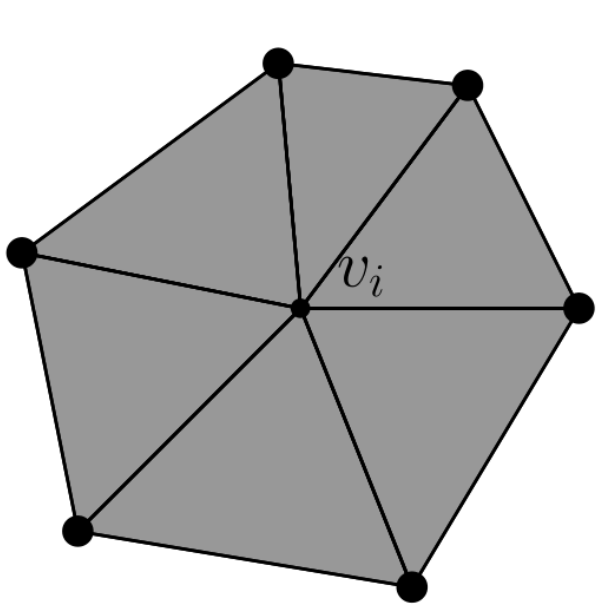}}
  \subfloat[]{\label{fig:neighborhood-d}\includegraphics[width=0.25\textwidth]{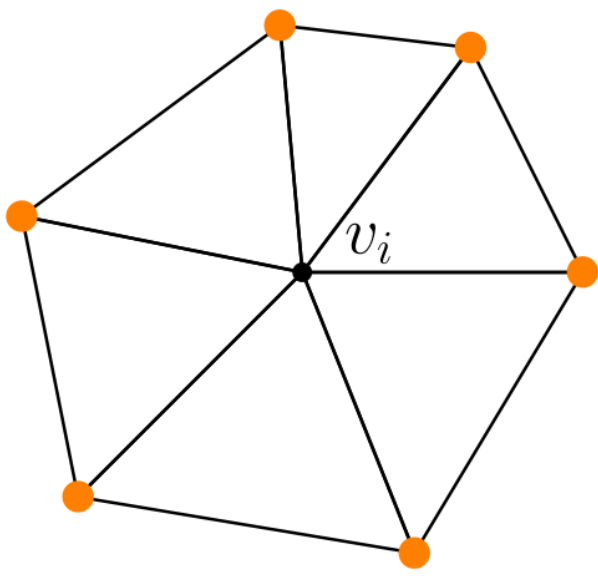}}
  \caption{The illustration of $D_{1}(\tau_i), B_{1}(\tau_i), B_{2}(\tau_i), M_{1}(v_i)$ and $N_{1}(v_i)$.
  The elements contained in $D_{1}(\tau_i), B_{1}(\tau_i), B_{2}(\tau_i), M_{1}(v_i)$ and $N_{1}(v_i)$ are plotted in \emph{green}, \emph{cyan}, \emph{blue}, \emph{gray} and \emph{orange} respectively.
  (a) $D_{1}(\tau_i)$ is the 1-ring of the triangle $\tau_i$, which refers to $3$ triangles;
  (b) $B_1(\tau_i)$ is the set of lines connecting the barycenter and vertices of $\tau_i$, which refers to $3$ lines.
  $B_2(\tau_i)$ is the set of lines connecting vertices of $\tau_i$ and barycenters of triangles contained in $D_{1}(\tau_i)$, which refers to $6$ lines;
  (c) $M_{1}(v_i)$ is the 1-disk of the vertex $v_i$, which refers to $6$ triangles;
  (d) $N_{1}(v_i)$ is the 1-neighborhood of $v_i$, which refers to $6$ vertices.}
  \label{fig:neighborhood}
\end{figure}

%\begin{figure}[htbp]
%  \centering
%  \includegraphics[width=1.0\linewidth]{figs/neighborhood}
%  \caption{The illustration of $D_{1}(\tau_i), B_{1}(\tau_i), B_{2}(\tau_i), M_{1}(v_i)$ and $N_{1}(v_i)$.
%  The elements contained in $D_{1}(\tau_i), B_{1}(\tau_i), B_{2}(\tau_i), M_{1}(v_i)$ and $N_{1}(v_i)$ are plotted in \emph{green}, \emph{cyan}, \emph{blue}, \emph{gray} and \emph{orange} respectively.
%  (a) $D_{1}(\tau_i)$ is the 1-ring of the triangle $\tau_i$, which refers to $3$ triangles;
%  (b) $B_1(\tau_i)$ is the set of lines connecting the barycenter and vertices of $\tau_i$, which refers to $3$ lines.
%  $B_2(\tau_i)$ is the set of lines connecting vertices of $\tau_i$ and barycenters of triangles contained in $D_{1}(\tau_i)$, which refers to $6$ lines;
%  (c) $M_{1}(v_i)$ is the 1-disk of the vertex $v_i$, which refers to $6$ triangles;
%  (d) $N_{1}(v_i)$ is the 1-neighborhood of $v_i$, which refers to $6$ vertices.
%  \label{fig:neighborhood}.
%  }
%\end{figure}

We further introduce the relative orientation of an edge $e$ to a triangle $\tau$, which is denoted by $\mathrm{sgn}(e,\tau)$ as
follows. First, we assume that all triangles are with counterclockwise orientation and all edges are with randomly chosen fixed orientations.
For an edge $e\prec\tau$, if the orientation of $e$ is consistent with the orientation of $\tau$, then $\mathrm{sgn}(e,\tau)=1$; otherwise $\mathrm{sgn}(e,\tau)=-1$.

\subsection{Piecewise Constant Function Spaces and Operators}
To describe piecewise constant data field, we present the concept of piecewise constant function space.
%The piecewise constant function space is a basic finite element space in numerical PDE.
Compared to piecewise linear function space which is suitable to deal with vertex-based problems, we find that, for feature preserving geometry processing \cite{Zhang:15,liu2017geodesic}, the piecewise constant function space is more suitable which is related to piecewise constant finite element method in numerical PDE.
For normal-based triangulated surface denoising, the piecewise linear function space requires the input to be vertex normals, while the input of the piecewise constant function space is face normals.
The vertex normals are averaged from face normals.
The second order geometry information of this smoothed vertex normal field is mush less sparse than that of the face normal field.
Thus, it is more appropriate to discretize our high order regularization model in the piecewise constant function space for preserving sharp features.
%We should point out that, the piecewise constant function spaces and associated first order difference operators were already introduced in our previous works \cite{liu2017geodesic} for computing geodesic curvature flow and \cite{Zhang:15} for denoising triangulated surfaces using TV method.
We should point out that, over triangulated surfaces, the second order difference operator is newly defined in this paper.

We denote the space ${V}_{M} = \mathbb{R}^{\mathrm{T}}$, which is
isomorphic to the piecewise constant function space over $M$.
$u = (u_{0}, u_{1}, \cdots, u_{\mathrm{T}-1})\in {V}_{M}$
means that the value of $u$ restricted on the triangle $\tau$ is
$u_{\tau}$, which is written as $u|_{\tau}$ sometimes.
%\rcomment{It is better to use $V_{E}, V_{T}$ as function space on edges and faces}
The inner product and norm in $V_{M}$ are as follows:
\begin{equation}\label{FacesInner}
(u^{1},u^{2})_{{V}_{M}}=\sum\limits_{\tau}u^{1}|_{\tau}u^{2}|_{\tau}s_{\tau}, \ \
\|u\|_{{V}_{M}} = \sqrt{(u,u)_{{V}_{M}}}, \forall u^{1},u^{2},u \in {V}_{M},
\end{equation}
where $s_{\tau}$ is the area of triangle $\tau$.
For any $u\in {V}_{M}$, the jump of $u$ over an edge $e$ is
\begin{eqnarray}\label{Jump} [u]_{e}=\left\{\begin{array}{rl}
\sum\limits_{\mathclap{\tau,e\prec\tau}} u|_{\tau}\mathrm{sgn}(e,\tau),  & e \not\subset \partial M\\
0,  &  e \subset \partial M
\end{array}\right..
\end{eqnarray}
It is then natural to define the gradient operator by
%\rcomment{write the domain and range when define the operator}
\begin{equation} \label{gradientOperator}
\nabla: V_M \rightarrow Q_M,  u \mapsto \nabla u, \ \  \nabla u|_e=[u]_e,\ \forall e, \  \mathrm{for} \ u\in V_{M},
\end{equation}
where $Q_M=\mathbb{R}^{\mathrm{E}}$ is the range of $\nabla$.
%where $Q_M$ is the range of $\nabla$, i.e.,
%$Q_{M} = \mathrm{Range}(\nabla)$.
The $Q_{M}$ space is equipped with the following inner product and norm:

\begin{equation} %\label{EdgesInner}
(q^{1},q^{2})_{Q_{M}}=\sum\limits_{e}q^{1}|_eq^{2}|_e \mathrm{len}(e), \ \ \|q\|_{Q_{M}} = \sqrt{(q,q)_{Q_{M}}},
\end{equation}
for $q^{1},q^{2}, q \in Q_{M}$, where $\mathrm{len}(e)$ is the length of the edge $e$.

%As the adjoint operator of $-\nabla$, the divergence operator is defined as
%\begin{equation*}
%  \mathrm{div}: Q_M \rightarrow V_M, q \mapsto \mathrm{div}q, \  \mathrm{for} \ q\in Q_{M}.
%\end{equation*}
%By using the above inner products in $V_{M}$ and $Q_{M}$, $\mathrm{div}q$ is given by
%\begin{equation}\label{DivOperator}
%(\mathrm{div}
%q)|_{\tau}= -\frac{1}{s_{\tau}}\sum\limits_{e\prec\tau} q|_e \mathrm{sgn}(e,\tau) \mathrm{len}(e), \forall \tau.
%\end{equation}

It is straightforward to derive the adjoint operator of $-\nabla$, the divergence operator $\mathrm{div}:\ Q_{M} \rightarrow V_{M}, q \mapsto \mathrm{div}q$, using the above inner products in $V_{M}$ and $Q_{M}$. For $q\in Q_{M}$, $\mathrm{div}q$ is given by
\begin{equation}\label{DivOperator}
(\mathrm{div}
q)|_{\tau}= -\frac{1}{s_{\tau}}\sum\limits_{\begin{subarray}{c}  e\prec\tau,  \\  e \not\subset \partial M  \end{subarray}}
%{\substack{e\prec\tau, \\ e \not\in \partial M} }
%{e\prec\tau, e \not\in \partial M}
q|_e \mathrm{sgn}(e,\tau) \mathrm{len}(e), \forall \tau.
\end{equation}

%Moreover, using the gradient and divergence operator mentioned above, the laplace operator $\Delta:V_{M} \rightarrow V_{M}, u \mapsto \Delta u$, can be derived as:
%%For $u\in V_{M}$, $\Delta u$ is given by
% \begin{equation}\label{LaplaceOperator}
%   (\Delta u) |_{\tau} =  -\frac{1}{s_{\tau}}
%   \sum\limits_{\begin{subarray}{c}  e \prec \tau,  \tau \cap \tau_{i} = e, \\ e \not\in \partial M   \end{subarray}}
%   %{\substack{e \prec \tau,  \tau \cap \tau_{i} = e, \\ e \not\in \partial M}}
%%   {\mathclap{\mbox{\tiny$\begin{array}{c}
%%{e, e\prec\tau},\\
%%\tau \cap \tau_{i} = e \end{array}$}}}
% (u_{\tau} - u_{\tau_i}) \mathrm{len}(e), \forall \tau.
% \end{equation}

%\begin{wrapfigure}{r}{0.4\textwidth}
%  \vspace{-20pt}
%  \begin{center}
%    \includegraphics{figs/2gradient}%[width=0.5\linewidth]%{figs/2gradient}
%  \end{center}
%  \vspace{-20pt}
%  \caption{The illustration of the jump of $[u]_e$ over the line $l$ plotted in \emph{cyan}
%  in triangle $\tau$
%  %with the barycenter plotted in \emph{red}.
%  \label{fig:2gradient}}
%  %\vspace{-10pt}
%\end{wrapfigure}

\begin{figure}[htbp]
  \centering
  \includegraphics[width=0.3\linewidth]{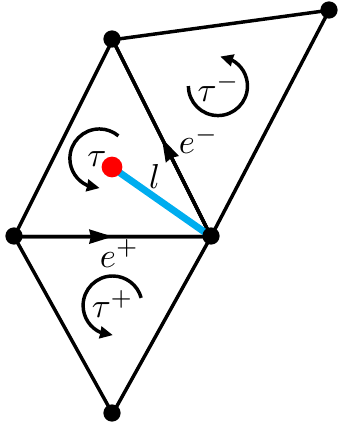}
  \caption{The illustration of the jump of $[u]_e$ over the line $l$ plotted in \emph{cyan}
  in triangle $\tau$
  with the barycenter plotted in \emph{red}.
  \label{fig:2gradient}}
\end{figure}

In the following section, we define a second order difference operator,
which will be used to construct our high order regularization models.
Let $l$ be a line connecting the barycenter and one vertex of the triangle $\tau$.
%Assume that there is a line $l$ in triangle $\tau$, which connects the barycenter and one vertex of $\tau$.
The two edges of triangle $\tau$, which share the common vertex of $l$, are denoted as $e^+$ and $e^-$ respectively.
%We assume orientations of these two edges are consistent with the orientation of $\tau$.
The two triangles, contained in $D_1(\tau_i)$ and share these two edges, are denoted as $\tau^+$ and $\tau^-$ respectively. %see Fig.~\ref{fig:2gradient}.
All aforementioned descriptions are indicated in \cref{fig:2gradient}.
%The triangle shares the edge $e+$ of triangle $\tau$ is denoted as $\tau +$, and that shares the edge $e^-$ is denoted as $\tau^-$.
%\rcomment{One of my concern is that readers will ask the motivation of defining such a operator. Motivation and explanation is needed}

We then define the jump of $[u]_e$ over the line $l$ in $\tau$ as
\begin{equation}\label{jumpFunction}
\begin{aligned}
 {[[u]]}_l =& [u]_{e^+}\mathrm{sgn}(e^+, \tau^+) + [u]_{e^-}\mathrm{sgn}(e^-,\tau^-) \\
           =& \big((u_{\tau}\mathrm{sgn}(e^+, \tau) + u_{\tau^+}\mathrm{sgn}(e^+, \tau^+))\mathrm{sgn}(e^+, \tau^+)\big) +  \\
            & \big((u_{\tau}\mathrm{sgn}(e^-, \tau) + u_{\tau^-}\mathrm{sgn}(e^-, \tau^-))\mathrm{sgn}(e^-, \tau^-)\big)  \\
           =& \big(u_{\tau^+} - u_{\tau}\big) + \big(u_{\tau^-} - u_{\tau}\big) \\
           =& u_{\tau^+} + u_{\tau^-} - 2u_{\tau},
\end{aligned}
\end{equation}
which is written as ${[[u]]}_{l,\tau}$ sometimes.
With Neumann boundary condition, we actually have, for any $u \in V_M$,
\begin{equation}\label{lJump}
[[u]]_{l}=
\left\{\begin{array}{rl}
u_{\tau^+} + u_{\tau^-} - 2u_{\tau}  ,& e^{+} \text{ or } e^{-} \not\subset \partial M\\
0,  &  e^{+} \text{ or } e^{-} \subset \partial M
\end{array}\right..
\end{equation}
From equation \cref{jumpFunction}, we can obviously see that,
the definition of $[[u]]_l$ is invariant under the choice of the orientation of edge $e$.

%\noindent $\textbf{\textit{Remark.}}$ The definition of $[[u]]_l$ is invariant under the choice of the orientation of edge $e$.

Then, the second order difference operator is defined by %\rcomment{Do not call it 2nd order gradient}
\begin{equation} \label{2ndOperator}
    \nabla^{2}:V_{M} \rightarrow P_{M} , u \mapsto \nabla^{2} u, \ \ \nabla^{2}u|_{\tau} = ( {[[u]]_{l_{0,\tau}}},  {[[u]]_{l_{1,\tau}}}, {[[u]]_{l_{2,\tau}}}),  \  \forall \tau, \ \mathrm{for} \ u \in V_M,
\end{equation}
where $P_M=\mathbb{R}^{\mathrm{T}} \times \mathbb{R}^{\mathrm{T}} \times \mathbb{R}^{\mathrm{T}}$ is the range of $\nabla^2$.
%where $P_M$ is the range of $\nabla^2$, i.e.,
%$P_{M} = \mathrm{Range}(\nabla^2)$.
The $P_{M}$ space is equipped with the following inner product and norm:

\begin{equation}\label{EdgesInner}
(p^{1},p^{2})_{P_{M}}= \sum\limits_{\tau} \sum\limits_{l \in B_{1}(\tau)}p^{1}|_l p^{2}|_l \mathrm{len}(l), \ \ \|p\|_{P_{M}} = \sqrt{(p,p)_{P_{M}}},
\end{equation}
for $p^{1},p^{2}, p \in P_{M}$, where $\mathrm{len}(l)$ is the length of line $l$.

\begin{lemma} \label{lemma1}
The adjoint operator of $\nabla^2$ that, $(\nabla^2)^{\star}: P_{M} \rightarrow V_{M}, p \mapsto (\nabla^2)^{\star}p$, has the following form:

\begin{equation}\label{Div2Operator}
((\nabla^2)^{\star}p)|_{\tau} = \frac{1}{s_{\tau}}\big(\sum\limits_
%{\small \mathclap  {\begin{subarray}{c}
%l \in B_1{(\tau)},\\
%e^{+} \text{ or } e^{-} \not\in \partial M
%\end{subarray}}}
{\mathclap{\mbox{\tiny $\begin{array}{c}
l \in B_2{(\tau)},\\
e^{+} \mathrm{or } \ e^{-} \not\subset \partial M\end{array}$}}}
p|_l \mathrm{len}(l) -
\sum\limits_
{\mathclap{\mbox{\tiny $\begin{array}{c}
l \in B_{1}(\tau),\\
e^{+} \mathrm{or} \ e^{-} \not\subset \partial M\end{array}$}}}2 p|_l \mathrm{len}(l)\big), \ \forall \tau, \ \mathrm{for} \ p \in P_M.
\end{equation}

\end{lemma}

\begin{proof}
As the definition of adjoint operator, we have
\begin{equation} \label{adjointoperator}
  \langle \nabla^{2}u, p \rangle_{P_M} = \langle u, (\nabla^2)^\star p \rangle_{V_M}.
\end{equation}
By using the inner products \cref{EdgesInner} and \cref{FacesInner} in $P_M$ and $V_M$, \cref{adjointoperator} can be reformulated as
\begin{equation} \label{innerAdjointOperator}
  \sum\limits_{\tau} \sum\limits_{l \in B_1{(\tau)}} [[u]]_{l} p|_l \mathrm{len}(l)  = \sum\limits_{\tau}{ u_{\tau} ((\nabla^2)^\star p) |_{\tau} s_{\tau}}.
\end{equation}
By using \cref{lJump}, the left-hand side of \cref{innerAdjointOperator} is actually
\begin{equation*}
\begin{aligned}
\sum\limits_{\tau} \sum\limits_{l \in B_1{(\tau)}}{[[u]]_{l} p|_l \mathrm{len}(l)}
= & \sum\limits_{\tau} \big( \sum\limits_{\mathclap
{\mbox{\tiny$\begin{array}{c}
l \in B_1{(\tau)}, \\
e^{+} \text{or }e^{-}\not\subset\partial M\end{array}$}}} (u_{\tau^+} + u_{\tau^-} - 2u_{\tau} ) p|_l \mathrm{len}(l)  \big)\\
= & \sum\limits_{\tau}\big(  \sum\limits_{\mathclap{\mbox{\tiny$\begin{array}{c}
l \in B_2{(\tau)},\\
e^{+} \text{or }e^{-} \not\subset \partial M\end{array}$}}} u_{\tau} p|_l \mathrm{len}(l)- \sum\limits_{\mathclap{\mbox{\tiny$\begin{array}{c}
l \in B_1{(\tau)},\\
e^{+} \text{or }e^{-} \not\subset \partial M\end{array}$}}}2 u_{\tau} p|_l \mathrm{len}(l)\big) \\
= & \sum\limits_{\tau} u_{\tau} \big(\sum\limits_{\mathclap{\mbox{\tiny$\begin{array}{c}
l \in B_2{(\tau)},\\
e^{+} \text{or }e^{-} \not\subset \partial M\end{array}$}}} p|_l \mathrm{len}(l) -
\sum\limits_{\mathclap{\mbox{\tiny$\begin{array}{c}
l \in B_1{(\tau)},\\
e^{+} \text{or }e^{-} \not\subset \partial M\end{array}$}}} 2 p|_l \mathrm{len}(l)
\big).
\end{aligned}
\end{equation*}
Therefore, we have
\begin{equation*}
  \sum\limits_{\tau} u_{\tau} \big(\sum\limits_{\mathclap{\mbox{\tiny$\begin{array}{c}
l \in B_2{(\tau)},\\
e^{+} \text{or }e^{-} \not\subset \partial M\end{array}$}}} p|_l \mathrm{len}(l) -
 \sum\limits_{\mathclap{\mbox{\tiny$\begin{array}{c}
l \in B_1{(\tau)},\\
e^{+} \text{or} e^{-} \not\subset \partial M\end{array}$}}} 2 p|_l \mathrm{len}(l)
\big) = \sum\limits_{\tau}{ u_{\tau} \big( ((\nabla^2)^\star p) |_{\tau} s_{\tau} \big)}.
\end{equation*}
Then, the assertion follows this immediately.
\end{proof}

To handle vectorial data, we extend the above concepts to vectorial cases. Two vectorial spaces $\mathbf{V}_M, \mathbf{P}_M$ are as follows:
$$\mathbf{V}_{M} =\underbrace{ V_{M} \times \cdots \times V_M}\limits_{\mathfrak{N}}, \mathbf{P}_{M} =\underbrace{ P_M \times \cdots \times P_M}\limits_{\mathfrak{N}},$$
for $\mathfrak{N}$-channel data.
The inner products and norms in $\mathbf{V}_M$ and $\mathbf{P}_M$ are as follows:
$$(\mathbf{u}^{1},\mathbf{u}^{2})_{\mathbf{V}_M} = \sum\limits_{1\leq i \leq \mathfrak{N}}(u_{i}^{1},u_{i}^{2})_{V_M},
\ \  \|\mathbf{u}\|_{\mathbf{V}_{M}} =
\sqrt{(\mathbf{u},\mathbf{u})_{\mathbf{V}_{M}}},\mathbf{u}^{1},\mathbf{u}^{2},\mathbf{u} \in \mathbf{V}_{M}$$
$$(\mathbf{p}^{1},\mathbf{p}^{2})_{\mathbf{P}_M} =
\sum\limits_{1\leq i\leq \mathfrak{N}}(p_{i}^{1},p_{i}^{2})_{P_M}, \
\  \|\mathbf{p}\|_{\mathbf{P}_{M}} =
\sqrt{(\mathbf{p},\mathbf{p})_{\mathbf{P}_{M}}},\mathbf{p}^{1},\mathbf{p}^{2},\mathbf{p} \in \mathbf{P}_{M}.$$
We mention that $\nabla \mathbf{u},  \nabla^{2} \mathbf{u}$ and their adjoint operators can be computed channel by channel.

\subsection{Two Discrete High Order Models Over Triangulated Surfaces}
Assume $f \in V_M$ be an observed noisy scalar field on $M$.
A high order regularization model reads
\begin{equation}\label{hotvModel}
 \min \limits_{u\in V_M}\{E_{\mathrm{ho}}(u) =R_{\mathrm{ho}}(\nabla^{2} u)+\frac{\alpha}{2}\left\| u-f \right\|^2_{V_M}\},
\end{equation}
where
\begin{equation*}
  R_{\mathrm{ho}}(\nabla^{2} u) = \sum\limits_{\tau} \sum\limits_{l \in B_1{(\tau)}} \big| [[u]]_{l} \big| \mathrm{len}(l) =  \sum\limits_{l} \big| [[u]]_{l} \big| \mathrm{len}(l),
\end{equation*}
%\rcomment{HOTC might not be a good name, i delete it.}
is the second order variation of $u$, and $\alpha$ is a positive fidelity parameter.
The minimization problem \cref{hotvModel} has a unique minimizer under the assumption of no degenerate triangles on $M$.
%Note that the model \eqref{hotvModel} is a generalization of Lysaker et al. \cite{Lysaker2003Noise} in image denoising.

Our second discrete high order regularization model is for vector field denoising over surfaces. Suppose $\mathbf{f}=(f_{1},f_{2},...,f_{\mathfrak{N}})$ be an observed noisy vector field.
The vectorial version of \cref{hotvModel} reads
\begin{equation} \label{vhotvModel}
 \min \limits_{\mathbf{u}\in \mathbf{V}_M}\{E_{\mathrm{vho}}(\mathbf{u}) =R_{\mathrm{vho}}(\nabla^{2} \mathbf{u})+\frac{\alpha}{2}\left\| \mathbf{u}-\mathbf{f} \right\|^2_{\mathbf{V}_M}\},
\end{equation}
where
\begin{equation*}
R_{\mathrm{vho}}(\nabla^{2} \mathbf{u}) %= (\mathrm{HOTV)}(\mathbf{u})
= \sum\limits_{\tau} \sum\limits_{l \in B_1{(\tau)}}\Big(\sum\limits_{i=1}^{\mathfrak{N}} \big| [[u_i]]_{l} \big||^{2}\Big)^{\frac{1}{2}}\mathrm{len}(l)
= \sum\limits_{l}\Big(\sum\limits_{i=1}^{\mathfrak{N}} \big| [[u_i]]_{l} \big||^{2}\Big)^{\frac{1}{2}}\mathrm{len}(l),
\end{equation*}
is the vectorial high order semi-norm.

%\begin{remark}
%The motivation of defining the second order difference operator will be explained as follows.
%First, we can see that discretizations of the second order difference operator \eqref{jumpFunction} and the laplace operator \eqref{LaplaceOperator} are totally different.
%%From \eqref{jumpFunction} and \eqref{LaplaceOperator}, we can see that discretizations of the second order difference operator and the laplace operator are different.
%For one triangle $\tau$, the second order difference operator $\nabla^2$ defined on one line $l$ can be seen as an analogue of one second partial derivative of the laplace operator $\triangle=u_{xx}+u_{yy}$ defined on 2D domain.
%Then, the second order difference operator is used to define our high order model \eqref{hotvModel}, which can be seen as a generalization of Lysaker et al. \cite{Lysaker2003Noise} in image denoising.
%%Considering of the sampling irregularity of triangulated surfaces, these models can be seen as anisotropy high order models.
%Moreover, the vectorial high order model \eqref{vhotvModel} can be used to design our normal filtering model \eqref{InitHOTVmodel}, which will be introduced in next section.
%At last, the second order difference operator \eqref{jumpFunction} can make our high order models more effective for preserving sharp features, while the laplace operator \eqref{LaplaceOperator} cannot.
%\end{remark}

\subsection{A Discussion On the Second Order Difference and Laplace Operator}
The Laplacian is the mostly frequently used high order operator in geometry processing.
Of course, it can also be used to construct high order regularization model.
In this subsection, we discuss differences between the second order difference operator \cref{2ndOperator} and Laplace operator in piecewise constant function space over triangulated surfaces.
The high order regularization models using these two high order operators are also compared.

For clarity, we firstly give the discretization of Laplace operator in piecewise constant function space.
By using the gradient operator \cref{gradientOperator} and divergence operator \cref{DivOperator}, the Laplace operator $\mathrm{\Delta}:V_{M} \rightarrow V_{M}, u \mapsto \mathrm{\Delta}u$, can be derived as:
%For $u\in V_{M}$, $\Delta u$ is given by
 \begin{equation}\label{LaplaceOperator}
   \mathrm{\Delta} u |_{\tau} =  -\frac{1}{s_{\tau}}
   \sum\limits_{\begin{subarray}{c}  e \prec \tau,  \tau \cap \tau_{i} = e, \\ e \not\subset \partial M   \end{subarray}}
   %{\substack{e \prec \tau,  \tau \cap \tau_{i} = e, \\ e \not\in \partial M}}
%   {\mathclap{\mbox{\tiny$\begin{array}{c}
%{e, e\prec\tau},\\
%\tau \cap \tau_{i} = e \end{array}$}}}
 (u_{\tau} - u_{\tau_i}) \mathrm{len}(e), \forall \tau.
 \end{equation}
For a noisy scalar field, an $\ell_1$-norm Laplacian regularization model reads
\begin{equation}\label{laplacianModel}
 \min \limits_{u\in V_M}\{E_{\mathrm{lap}}(u) =R_{\mathrm{lap}}(\mathrm{\Delta} u)+\frac{\alpha}{2}\left\| u-f \right\|^2_{V_M}\},
\end{equation}
where the regularization term is defined as
\begin{equation*}
  R_{\mathrm{lap}}(\mathrm{\Delta} u) =  \sum\limits_{\tau} \big| \mathrm{\Delta} u|_{\tau}  \big| s_{\tau}.
\end{equation*}
%This regularization term can be seen as an analogue to
%\begin{equation*}
%  \int_{\mathrm{\Omega}} |u_{xx} +  u_{yy}| dx dy
%\end{equation*}
%in 2D domain.

As we can see, discretizations of the second order operator \cref{2ndOperator} and Laplace operator \cref{LaplaceOperator} are totally different.
Our second order operator can be seen as a set of second order central differences defined over $l$ along three different directions in one triangle.
%In other words, the second order operator \eqref{2ndOperator} can be seen as a set of second partial derivatives of the laplace operator in 2D domain.
The second order regularization term of high order model \cref{hotvModel} is an analogue of the regularization term that
\begin{equation*}
  \int_{\mathrm{\Omega}} |u_{xx}| +  |u_{yy}| dx dy
\end{equation*}
in 2D domain, while the Laplacian regularization term of model \cref{laplacianModel} can be seen as an analogue to
\begin{equation*}
  \int_{\mathrm{\Omega}} |u_{xx} +  u_{yy}| dx dy.
\end{equation*}
The regularization term using second order operator \cref{2ndOperator} can be computed separately in different directions, while that using the Laplace operator \cref{LaplaceOperator} cannot.
Compared to the $\ell_1$-norm Laplacian model \cref{laplacianModel}, the second order regularization model \cref{hotvModel} is more effective for recovering sharp signals over triangulated surfaces; see the comparison of vectorial implementations of these two models in \cref{fig:secondOrderAndLaplaceOperator}.
Moreover, the second order regularization model overcomes the staircase effect introduced by first order models.

\section{Normal Filtering using High Order Model with Dynamic Weights}
\label{sec:normalFiltering}
The recent TV \cite{Zhang:15} and $\ell_0$ \cite{He13} based minimization methods use the concept of sparsity of first order information to remove noise from triangulated surfaces.
These methods preserve sharp features well, but suffer from the staircase effect in smooth regions inevitably.
In this section, a high order normal filtering model with dynamic weighs is proposed for preserving sharp features and removing
the staircase effect in smooth regions simultaneously.
The dynamic weights are applied in the proposed model to significantly improve effectiveness for preserving sharp features.
\subsection{High Order Normal Filtering Model with Dynamic Weights} \label{sec:HOTVModel}
For a given noisy surface $M^{in}$, we write the face normals as $\mathbf{N}^{in}$.
To remove noise in $\mathbf{N}^{in}$ through our vectorial high order model \cref{vhotvModel} with multiple spherical constraints, we propose the following variational model:

\begin{equation} \label{InitHOTVmodel}
\begin{aligned}
\min \limits_{\textbf{N}\in C_\textbf{N}}\{E(\textbf{N}) = R_{\mathrm{vhow}}(\nabla ^ 2 \textbf{N})+\frac{\alpha}{2}\left\|\textbf{N}-\textbf{N}^{in}\right\|^2_{\mathbf{V}_M}\},
\end{aligned}
\end{equation}
where
\begin{align}
C_\textbf{N}                        &=      \{\textbf{N} \in \mathbf{V}_M:\left\|\textbf{N}_\tau\right\|=1,\ \forall\tau\},\notag\\
R_{\mathrm{vhow}}(\nabla ^ 2 \textbf{N})     &=      \sum\limits_{l}w_l{\Big(\sum\limits_{i=1}\limits^3(\nabla^2 N_i |_l)^{2}\Big)^{\frac{1}{2}}}\mathrm{len}(l). \notag
\end{align}
Note that $\mathbf{V}_{M}$ denotes $3$-channel $V_M$ here.
%\rcomment{clarify that 3 channels are used for $V_M$}
%We use high order model, which is known to recover smooth surfaces, as basic ingredient.
%However, compared to TV based models,
%the high order model has the major challenge for recovering the sharp features.
%Thus, we use the dynamic weights to enhance the sparsity of the high order model for
%improving the sharp features reconstruction.
%In general, this combination is able to preserve sharp features well and at the same time avoid the staircase effect in smooth regions.
The dynamic weight $w_l$ on each $l$ of triangle is defined as
\begin{equation}\label{eqn:weightFunction}
  %w_l = \exp(-\left\|\nabla^2 \mathbf{N} |_l\right\|^4).
  w_l = \exp(-\left\| \mathbf{N}_{\tau^{+}} + \mathbf{N}_{\tau^{-}} - 2 \mathbf{N}_{{\tau}}  \right\|^4).
\end{equation}
The weight $w_l$ is designed to monotonically decrease with respect to the absolute second order difference defined on $l$.
%In addition, this weight needs to be dynamically updated according to the face normals.
%In addition, we use an exponent of $4$ instead of the usual $2$ in \eqref{eqn:weightFunction} for a more sparsity enhancement.
%The relations between $l$ and $\tau_+$, $\tau_-$, $\tau$ are already shown in Fig.\ref{fig:2gradient}.

For most surfaces, the proposed vectorial high order regularization model \cref{vhotvModel} can achieve good denoising results.
However, in rare cases, where the noise level is increased, the proposed model \cref{vhotvModel} may smooth some sharp features a little.
%We use vectorial high order regularization model \eqref{vhotvModel}, which is known to recover smooth surfaces, as basic ingredient.
%However, in rare cases, our high order model cannot compete with TV based method for recovering sharp features.
%For example, when the noise level is increased, our high order model may smooth some sharp features a little.
%our high order model has the major challenge for recovering the sharp features.
Thus, we use the dynamic weights, updated with respect to the face normals in each iteration, to enhance the sparsity of our high order model for improving sharp features reconstruction.
%The weight $w_l$ is designed to monotonically decrease with respect to the absolute second order difference defined on $l$.
%In addition, we use an exponent of $4$ instead of the usual $2$ in \eqref{eqn:weightFunction} for a more sparsity enhancement.
%The weight $w_l$ needs to be dynamically updated according to the face normals.
The dynamic weights scheme is inspired by Cand\`{e}s et al. in \cite{Cand2008}.
%which can be used to enhance the sparsity of signal and complement $\ell_1$ norm.
These dynamic weights penalize smooth regions (smoothly curved regions and flat regions) more than sharp features, which can be applied to achieve the lower-than-$\ell_1$-sparsity effect.
%for significantly improving effectiveness for preserving sharp features.
In general, the combination of the high order model and the dynamic weights is able to preserve sharp features well and at the same time
recover smooth regions without staircase effects.
%avoid the staircase effect in smooth regions.

%\begin{remark}
%We use high order model, which is known to recover smooth surfaces, as basic ingredient.
%However, compared to TV based models,
%the high order model has the major challenge for recovering the sharp features.
%Thus, we use the dynamic weights to enhance the sparsity of the high order model for
%improving the sharp features reconstruction.
%In general, this combination is able to preserve sharp features well and at the same time avoid the staircase effect in smooth regions.
%\end{remark}

%Compared to reweighted $L_1$ minimization \eqref{eqn:reweightedL1},
%which first minimize one problem with fixed weights then updates the weights to re-solve the problem,
%our method is more efficient (because we no need to solve the problem exactly with fixed weights).

%The role of this weight can achieve the effect of lower-than-$L_1$-sparsity.
%This weight need to be iteratively updated according to face normals.
%And this reweighted $L_1$ scheme can improve the quality of recovering sharp features
%for penalizing these features less than smooth regions of the surface.
%
%This theoretical observation first proposed by Candes et al. \cite{Cand2008} that reweighting within the $L_1$ realm should enhance the sparsity
%and improve the quality of the reconstruction.
%This reweighted  method was used in \cite{Avron2010} and \cite{Zhang:15} for point cloud and mesh denoising respectively, which can improve preserving sharp features
%for penalizing these features less than smooth regions of the surface.

\subsection{Augmented Lagrangian Method for Solving the High Order Normal Filtering Model} \label{sec:alm}
It is challenging to solve the normal filtering model \cref{InitHOTVmodel} due to the non-differentiability and nonlinear constraints.
Recently, the variable splitting and augmented Lagrangian method (ALM) have attained intensive attention for their efficiency in many $\ell_1$ related optimization problems \cite{Oden1992Augmented,Osher2005An,Wu2010Augmented}.
Hence, we introduce an auxiliary variable and use ALM to handle the regularization term of \cref{InitHOTVmodel}.
Moreover, in each iteration of ALM, the weight \cref{eqn:weightFunction} is updated dynamically.

%\rcomment{We should mention the method fix weight first and then update weight}

We first introduce a new variable $\textbf{p}\in \mathbf{P}_M$ and rewrite the problem \cref{InitHOTVmodel} as

\begin{equation}\label{HOTVmodel}
\begin{aligned}
\min \limits_{\textbf{N}\in \mathbf{V}_M,\textbf{p}\in \mathbf{P}_M}\{& R_{\mathrm{vhow}}(\textbf{p})+  \frac{\alpha}{2}\left\|\textbf{N}-\textbf{N}^{in}\right\|^2_{\mathbf{V}_M}
+ \sigma_{C_{\mathbf{N}}}(\textbf{N})
\} \\
&\textmd{s.t.} \quad \textbf{p} = \nabla ^2 \textbf{N},
\end{aligned}
\end{equation}
where
\begin{eqnarray*}
\sigma_{C_{\mathbf{N}}}(\textbf{N})= \left
\{ \begin{array}{rl}
0,         & \mathbf{N} \in  C_{\mathbf{N}} \\
+\infty,   & \mathbf{N} \notin  C_{\mathbf{N}}.
\end{array}
\right.
\end{eqnarray*}

Accordingly, we define the following augmented Lagrangian function

\begin{equation}\label{ALGfunction}
\begin{aligned}
  L(\textbf{N}, \textbf{p}; \lambda_\mathbf{p}) = & R_{\mathrm{vhow}}(\textbf{p})+ \frac{\alpha}{2}\left\|\textbf{N}-\textbf{N}^{in}\right\|^2_{\mathbf{V}_M} + \sigma_{C_{\mathbf{N}}}(\textbf{N}) \\ & + {(\lambda_\mathbf{p}, \textbf{p} - \nabla ^2 \textbf{N})}_{\mathbf{P}_{M}} + \frac{r_\mathbf{p}}{2}\|\textbf{p} - \nabla ^2 \textbf{N}\|^2_{\mathbf{P}_{M}},
\end{aligned}
\end{equation}
where $\lambda_\mathbf{p}$ is a Lagrange multiplier and $r_\mathbf{p}$ is a positive real number.
This primal variables update procedure can be separated into two subproblems:

\begin{itemize}
  \item The $\mathbf{N}$-sub problem: given $\textbf{p}$
  \begin{equation}\label{eqn:n-sub}
    \min\limits_{\mathbf{N} \in \mathbf{V}_{M}}\frac{\alpha}{2}\|\mathbf{N} - \mathbf{N}^{in}\|^2_{\mathbf{V}_M} + (\lambda_\mathbf{p}, -\nabla^{2}\mathbf{N})_{\mathbf{P}_M} + \frac{r_\mathbf{p}}{2}\|\mathbf{p} - \nabla^{2}{\mathbf{N}} \|^{2}_{\mathbf{P}_M} + \sigma_{C_{\mathbf{N}}}(\textbf{N});
  \end{equation}

%  \item The $\textbf{Z}$-sub problem: given $\textbf{N}$
%  \begin{equation}\label{eqn:z-sub}
%    \min\limits_{\mathbf{Z} \in \mathbf{V}_M} \sigma(\textbf{Z}) + (\lambda_\mathbf{Z}, \mathbf{Z})_{\mathbf{V}_M} + \frac{r_\mathbf{Z}}{2}\|\mathbf{Z}-\mathbf{N} \|^2_{\mathbf{V}_M}.
%  \end{equation}

  \item The $\mathbf{p}$-sub problem: given $\textbf{N}$
  \begin{equation}\label{eqn:p-sub}
    \min\limits_{\mathbf{p} \in \mathbf{P}_M} R_{\mathrm{vhow}}(\mathbf{p}) + (\lambda_\mathbf{p}, \mathbf{p})_{\mathbf{P}_M} + \frac{r_\mathbf{p}}{2}\|\mathbf{p}-\nabla^{2}\mathbf{N} \|^2_{\mathbf{P}_M}.
  \end{equation}
\end{itemize}

%\subsubsection{Solving the $\mathbf{N}$-sub Problem} \label{sec:solve-n}
The $\mathbf{N}$-sub problem is a quadratic minimization with orthogonality constraints.
An iterative method is needed to find its exact solution \cite{Goldfarb2009A,Wen2013A,Lai2014A}.
% e.g., the SOC method in \cite{Lai2014A} and the curvilinear method in \cite{Goldfarb2009A,Wen2013A}.
Due to error forgetting and cancellation properties of ALM for $\ell_1$ minimization problem \cite{Yin2013Error}, the sub-optimization problem \cref{eqn:n-sub} does not have to be solved very accurately.
%Due to the augmented Lagrangian method can forget and cancel the error for $\ell_1$ minimization problem (see Yin and Osher's error forgetting paper \cite{Yin2013Error} for details), we do not need to solve this subproblem very accurately.
Here we adopt an approximate strategy to balance the precision and computational efficiency.
We first ignore $\sigma_{C_{\mathbf{N}}}(\textbf{N})$ and solve a quadratic programming and then project the minimizer to an unit sphere.
The quadratic problem (without constraints) has the first order optimality condition
\begin{equation}\label{eqn:optimalitycondition}
  r_\mathbf{p} ((\nabla^2)^{\star}\mathbf{\nabla^2N}) + \alpha \mathbf{N} = r_\mathbf{p} ((\nabla^2)^{\star}\mathbf{p}) + (\nabla^2)^{\star}\lambda_{\mathbf{p}}  + \alpha \mathbf{N}^{in}.
\end{equation}
%\begin{equation}\label{eqn:optimalitycondition}
%  r_\mathbf{p} ((\nabla^2)^{\star}\mathbf{\nabla^2N})|_\tau + \alpha \mathbf{N}_{\tau} = r_\mathbf{p} ((\nabla^2)^{\star}\mathbf{p})|_\tau + (\nabla^2)^{\star}\lambda_{\mathbf{p}}|_\tau  + \alpha \mathbf{N}^{in}_{\tau},
%\end{equation}
%which is in three channels.
%By \eqref{Div2Operator}, \eqref{eqn:optimalitycondition}
This equation can be reformulated into a sparse and positive semidefinite linear system, %with the same coefficient matrix,
which can be solved by various well-developed numerical packages.
Here we use conjugate gradient (CG) method to solve the problem.
The maximum number of iteration of CG method is set to be $10$ for efficiency.
Then, we directly project the solution onto the unit sphere.
%\rcomment{Perhaps mention that we should repeat this multiple steps, here only chose one step. Ans also mention augmented Lagrangian method can forget and cancel for l1 problem (see Yin-Osher's error forgetting paper), thus we do not need to solve this subproblem very accurately. In addition, we need to say how to solve the above equation. PS. Have you tried CG method we discussed before?}

%\subsubsection{Solving the $\mathbf{Z}$-sub Problem}
%The $\mathbf{Z}$-sub problem \eqref{eqn:z-sub} can be reformulated as
%  \begin{equation}\label{eqn:z-sub1}
%    \min\limits_{\mathbf{Z} \in \mathbf{V}_M} \sigma(\textbf{Z}) + \frac{r_\mathbf{Z}}{2}\|\mathbf{Z}-(\mathbf{N}-\frac{\lambda_\mathbf{Z}}{r_\mathbf{Z}} )\|^2_{\mathbf{V}_M},
%  \end{equation}
%which has the closed form solution
%\begin{equation}\label{eqn:z-solution}
%  \mathbf{Z} = \frac{\kappa}{\|\kappa\|},
%  %\mathbf{Z} = \frac{\mathbf{N}- \frac{\lambda_\mathbf{Z}}{r_\mathbf{Z}}}{\|\mathbf{N}- \frac{\lambda_\mathbf{Z}}{r_\mathbf{Z}}\|}.
%\end{equation}
%where
%\begin{equation*}
%  \kappa = \mathbf{N}- \frac{\lambda_\mathbf{Z}}{r_\mathbf{Z}}.
%\end{equation*}

Next, we solve the $\mathbf{p}$-sub problem \cref{eqn:p-sub}.
By \cref{EdgesInner}, this problem can be written as
\begin{equation}\label{eqn:p-sub0}
    \min\limits_{\mathbf{p}}\sum_{l} w_{l}|\mathbf{p}_l| \mathrm{len}(l) + \sum_{l}(\lambda_{p_l}, \mathbf{p}_l)\mathrm{len}(l) + \sum_{l}\frac{r_\mathbf{p}}{2} |\mathbf{p}_l - (\nabla^2 \mathbf{N})|_l|^2 \mathrm{len}(l).
\end{equation}
The problem \cref{eqn:p-sub0} is decoupled and can be solved line-by-line. For each line $l$ connecting the barycenter and one vertex of one triangle, we need to solve
\begin{equation*} %\label{eqn:p-sub1}
    \min\limits_{\mathbf{p}_l}w_{l}|\mathbf{p}_l| + (\lambda_{p_l}, \mathbf{p}_l) + \frac{r_\mathbf{p}}{2} |\mathbf{p}_l - (\nabla^2 \mathbf{N})|_l|^2,
\end{equation*}
which has a closed form solution

\begin{equation} \label{eqn:p-closedSolution}
\textbf{p}_l= \left
\{ \begin{array}{rl}
(1- \frac{1}{r_\mathbf{p} |\xi_l|})\xi_l,         & |\xi_l| > \frac{w_l}{r} \\
0,   & |\xi_l| \leq \frac{w_l}{r},
\end{array}
\right.
\end{equation}
with
\begin{equation*}
  \xi = \nabla^2 \mathbf{N} - \frac{\lambda_\mathbf{p}}{r_\mathbf{p}}.
\end{equation*}

In summary, the algorithm of high order normal filtering model \cref{InitHOTVmodel} is given in \Cref{alg:normalFiltering}.
Based on the variable splitting and ALM, this algorithm solves the non-differentiability problem with nonconvex constraints by iterating several simple operations.
We should point out that, in the conventional reweighted $\ell_1$ minimization \cref{eqn:reweightedL1}, the minimization problem with fixed weights is usually solved exactly.
Therefore, the reweighted strategy is time-consuming.
In contrast, \Cref{alg:normalFiltering} updates the weights in each iteration.
It can be regarded as an inexact but more efficient version of the conventional reweighted minimization algorithm.
Although we currently cannot give a rigorous proof of convergence for \Cref{alg:normalFiltering}, our numerical experiments strongly validate it in practice.
A theoretical analysis of this algorithm is worthy of the future research.

%Although Algorithm \ref{alg:normalFiltering} is currently lack of rigorous proof, our numerical experiments strongly validate it in practice.
%A rigorous mathematical validation of
%Algorithm \ref{alg:normalFiltering} will be certainly considered in our future work.

\begin{algorithm}
\caption{Augmented Lagrangian Method for High Order Normal Filtering with Dynamic Weights}
\label{alg:normalFiltering}
\begin{algorithmic}[1]
\State{\textbf{Initialization}: $\lambda^{0}_{\mathbf{p}} =0$, $\mathbf{N}^{-1}=0$, $\mathbf{P}^{-1}=0$, $k=0$, $K=100$, $\epsilon = 1e-4$}
\Do

\State{\textbf{1. Compute} $\mathbf{N}^{k}$ from \cref{eqn:optimalitycondition}, for fixed ($\lambda_\mathbf{p}^k$, $\mathbf{p}^{k-1}$); \textbf{Normalize} $\mathbf{N}^k$}

\State{\textbf{2. Compute} $\mathbf{p}^k$ from \cref{eqn:p-closedSolution}, for fixed $(\lambda^{k}_{\mathbf{p}}, \mathbf{N}^k)$}

\State{\textbf{3. Update} Lagrange multiplier $\lambda^{k+1}_{\mathbf{p}} =  \lambda^{k}_{\mathbf{p}} + r_{\mathbf{p}}(\mathbf{p}^k - \nabla^2 \mathbf{N}^k) $}

\State{\textbf{4. Update} each weight $w_l$ through \cref{eqn:weightFunction} with respect to $\mathbf{N}^{k}$}

\doWhile {$\|\mathbf{N}^{k} - \mathbf{N}^{k-1}\|_{\mathbf{V}_M} < \epsilon$ or $k>K$} \\
\Return $\mathbf{N}$
\end{algorithmic}
\end{algorithm}

\section{Folding Free Vertex Updating Method}
\label{sec:VertexUpdateModel}
%Two-stage methods, face normal filtering followed by vertex updating, are important triangulated surface denoising techniques \cite{Taubin2001Linear,Yagou2002Mesh,Shen2004Fuzzy,Lee2006Feature,Sun:07,Sun:08,Zheng:11,Zhang:15,Zhang2015Guided}. All these methods use almost the same vertex updating model, which originated from Taubin \cite{Taubin2001Linear}, and has a beautiful implementation by Sun et al. \cite{Sun:07}.
%It implicitly assumes that the surface corrupted by Gaussian noise along its normal direction and, the approach by Sun et al. \cite{Sun:07} achieves impressive results.
%However, in practice, noise usually corrupt the surface in random directions.
%In this situation, the updating scheme does not work well and suffers from foldovers.
%Moreover, large scale noise in random directions make the matter even worse.
%To address this issue, we propose a new vertex updating method to recover a folding free surface with respect to the filtered face normals, which has crucial advantages in the presence of noise in random directions. \rcomment{we should mention a two state method is used in this paper in the introduction part. Perhaps move some of this paragraph to the introduction part}

%\subsection{A Review of the Previous Vertex Updating Method}
%\label{sec:VertexUpdatePrevious}
%\rcomment{Perhaps no need to separate this section as two subsections, make the review part more condense}
After restoring the face normal field by \cref{alg:normalFiltering}, the positions of vertices need to be reconstructed to match the updated face normals.
As mentioned in \cref{sec:introduction}, all the existing two-stage methods \cite{Taubin2001Linear,Yagou2002Mesh,Shen2004Fuzzy,Lee2006Feature,Sun:07,Sun:08,Zheng:11,Zhang:15,Zhang2015Guided} use the same vertex updating model
\begin{equation}\label{eq:Sun}
  \min\limits_{v} \sum\limits_{\tau}\sum\limits_{(v_i, v_j) \in \tau} s_{\tau}(\mathbf{N}_{\tau} \cdot (v_i - v_j))^2,
\end{equation}
where $s_{\tau}$ is the area and $\mathbf{N}_{\tau}$ is the filtered normal of $\tau$.
The gradient descent method is used to minimize this optimization problem, where its initialization is the restored face normal field.
This optimization problem is to penalize the non-orthogonality between the filtered face normal and the three edges at each face over the surface.
%Using the gradient descent method to \eqref{eq:Sun}, the vertex updating was implemented by Sun et al. \cite{Sun:07} as
%\begin{equation}\label{eq:VertexUpdatingEqn}
%        v^{'}_{i} = v_{i} + \frac{1}{|M_{1}(i)|}\sum\limits_{\tau \in M_{1}(i)} \mathbf{N}_{\tau} (\mathbf{N}_{\tau} \cdot (c_{\tau} - v_i)),
%\end{equation}
%where $v^{'}_{i}$ is the updated vertex position of $v_i$, $|M_{1}(i)|$ is the number of triangles containing $v_i$ (see Fig.~\ref{fig:neighborhood}(c))
%and $c_{\tau}$ is the barycenter of $\tau$.
%The area weights of \eqref{eq:VertexUpdatingEqn} are all replaced by $1$ for simplicity and speed.
%When a mesh is corrupted by noise along its normal direction, this vertex updating method works pretty well. \rcomment{No need to mention the details of gradient descent. Perhaps delete the whole paragraph}
However, when a surface is corrupted by noise in random directions, the vertex updating method \cite{Sun:07} usually produces foldings, even with the exact (ground truth) face normals.
In addition, large scale noise make this phenomenon even worse; see the second column of \cref{fig:vertexUpdatingComparisons}.
The reason is that the model \cref{eq:Sun} only penalizes the non-orthogonality and can not distinguish $-\mathbf{N}_{\tau}$ and $\mathbf{N}_{\tau}$. Thus, the model neglects the orientations of triangle face normals and leads to updating ambiguities.
In other words,
%the method only requires the orthogonality between the triangle face and its corresponding filtered face normal, but cannot determine the orientation of the face.
a vertex $v_i$ of triangle $\tau$ may be updated along the direction $-\mathbf{N}_{\tau}$ instead of $\mathbf{N}_{\tau}$.
These triangle-wise orientation ambiguities cause inconsistent normal vectors crossing different triangles.

%\subsection{A New Vertex Updating Method}
%\label{sec:ANewVertexUpdatingMethod}
%To address the orientation ambiguity problem, we propose a new vertex updating method, which not only considers the orthogonality between the triangle face and its corresponding normal direction but also takes into account the orientation of the face.
%Our method has crucial advantages in the presence of noise in random directions.

To address the orientation ambiguity problem, we propose a new vertex updating method, which reconstructs the surface from a given normal vector field by solving the following minimization problem:
%\begin{equation}\label{eq:vertexUpdateModel}
%    \min\limits_{v} \{E(v) = -\sum\limits_
%    {\mathclap{\footnotesize \mbox{$\tau=(v_i,v_j,v_k)$}}}
%    %{\tau=(v_i,v_j,v_k)}
%{s_{\tau} \mathbf{N}}_{\tau}\cdot\left(\frac{(v_j-v_i)\times(v_k-v_i)}{\left\|(v_j-v_i)\times(v_k-v_i)\right\|}\right) + \frac{\eta}{2}\sum\limits_{v} || v - v^{in} ||^2 \},
%\end{equation}
\begin{equation}\label{eq:vertexUpdateModel}
    \min\limits_{v} \{E(v) = -\sum\limits_
    {\mathclap{\footnotesize \mbox{$\tau=(v_i,v_j,v_k)$}}}
    %{\tau=(v_i,v_j,v_k)}
{s_{\tau} \mathbf{N}}_{\tau}\cdot\left(\frac{(v_j-v_i)\times(v_k-v_i)}{\left\|(v_j-v_i)\times(v_k-v_i)\right\|}\right) + \frac{\eta}{2} || v - v^{in} ||^2 \},
\end{equation}
where $(v_i, v_j, v_k)$ are vertices of $\tau$ with counterclockwise order and
$\eta$ is a small positive parameter.
%We should explain more about this model.
The first term of \cref{eq:vertexUpdateModel} is used to solve the orientation ambiguity problem.
This term not only considers the orthogonality between the triangle face and its corresponding normal direction, but also takes into account the orientation of the face.
Thus, compared to \cref{eq:Sun}, the energy of model \cref{eq:vertexUpdateModel} poses no ambiguity.
The second term of \cref{eq:vertexUpdateModel} is a fidelity term.

The partial derivatives of the energy $E(v)$ with respect to $v_i$ is as follows:

\begin{equation*} %\label{eq:vertexParital}
\small
%\begin{split}
\begin{aligned}
\nabla_{v_i}E(v) &= \\
-\sum\limits_
{\mathclap{\mbox{\scriptsize
$\begin{array}{c}
\tau \in M_1(v_i),\\
\tau=(v_i,v_j,v_k)\end{array}$}}}
s_{\tau} & \bigg(\frac{\mathbf{N}_{\tau}\times(v_k-v_j)}{\left\|(v_j-v_i)\times(v_k-v_i)\right\|} +
                                    \frac{\mathbf{N}_{\tau}\times(v_k-v_j)}{\left\|(v_k-v_j)\times(v_i-v_j)\right\|} +
                                      \frac{\mathbf{N}_{\tau}\times(v_k-v_j)}{\left\|(v_i-v_k)\times(v_j-v_k)\right\|} \\    &-\frac{\mathbf{N}_{\tau}\cdot\left[(v_j-v_i)\times(v_k-v_j)\right]\left[(v_j-v_i)\times(v_k-v_i)\times(v_k-v_j)\right]}{{\left\|(v_j-v_i)\times(v_k-v_i)\right\|}^3}\\
                                      &-\frac{\mathbf{N}_{\tau}\cdot\left[(v_k-v_j)\times(v_i-v_j)\right]\left[(v_k-v_j)\times(v_i-v_j)\times(v_k-v_j)\right]}{{\left\|(v_k-v_j)\times(v_i-v_j)\right\|}^3}\\
                                      &-\frac{\mathbf{N}_{\tau}\cdot\left[(v_i-v_k)\times(v_j-v_k)\right]\left[(v_i-v_k)\times(v_j-v_k)\times(v_k-v_j)\right]}{{\left\|(v_i-v_k)\times(v_j-v_k)\right\|}^3}\bigg) +\eta(v_i-v_i^{in}).\\
%\end{split}
\end{aligned}
\end{equation*}
%\begin{equation*} %\label{eq:vertexParital}
%\small
%%\begin{split}
%\begin{aligned}
%\nabla_{v_i}E(v)= - \sum\limits_
%{\mathclap{\mbox{\scriptsize
%$\begin{array}{c}
%\tau \in M_1(v_i),\\
%\tau=(v_i,v_j,v_k)\end{array}$}}}
%s_{\tau} & \bigg(\frac{\mathbf{N}_{\tau}\times(v_k-v_j)}{\left\|(v_j-v_i)\times(v_k-v_i)\right\|} +
%                                    \frac{\mathbf{N}_{\tau}\times(v_k-v_j)}{\left\|(v_k-v_j)\times(v_i-v_j)\right\|} +
%                                      \frac{\mathbf{N}_{\tau}\times(v_k-v_j)}{\left\|(v_i-v_k)\times(v_j-v_k)\right\|} \\    &-\frac{\mathbf{N}_{\tau}\cdot\left[(v_j-v_i)\times(v_k-v_j)\right]\left[(v_j-v_i)\times(v_k-v_i)\times(v_k-v_j)\right]}{{\left\|(v_j-v_i)\times(v_k-v_i)\right\|}^3}\\
%                                      &-\frac{\mathbf{N}_{\tau}\cdot\left[(v_k-v_j)\times(v_i-v_j)\right]\left[(v_k-v_j)\times(v_i-v_j)\times(v_k-v_j)\right]}{{\left\|(v_k-v_j)\times(v_i-v_j)\right\|}^3}\\
%                                      &-\frac{\mathbf{N}_{\tau}\cdot\left[(v_i-v_k)\times(v_j-v_k)\right]\left[(v_i-v_k)\times(v_j-v_k)\times(v_k-v_j)\right]}{{\left\|(v_i-v_k)\times(v_j-v_k)\right\|}^3}\bigg) +\eta(v_i-v_i^{in}).\\
%%\end{split}
%\end{aligned}
%\end{equation*}
Using the two facts that
\begin{equation*}
\left\|(v_j-v_i)\times(v_k-v_i)\right\| = \left\|(v_k-v_j)\times(v_i-v_j)\right\| = \left\|(v_i-v_k)\times(v_j-v_k)\right\| = 2\mathcal{S}_{\tau},
\end{equation*}
\begin{equation*}
\frac{(v_j-v_i)\times(v_k-v_i)}{2\mathcal{S}_{\tau}} = \frac{(v_k-v_j)\times(v_i-v_j)}{2\mathcal{S}_{\tau}} = \frac{(v_i-v_k)\times(v_j-v_k)}{2\mathcal{S}_{\tau}} = \mathcal{N}_{\tau},
\end{equation*}
where $\mathcal{S}_{\tau}$ and $\mathcal{N}_{\tau}$ are updating area and normal of triangle $\tau$ according to the updated vertices $v$,
%with updating vertices $(v_i, v_j, v_k)$,
we arrive at
\begin{equation}\label{eq:6}
\nabla_{v_i}E(v) =  \sum\limits_
{\mathclap{\mbox{\scriptsize
$\begin{array}{c}
\tau \in M_1(v_i),\\
\tau=(v_i,v_j,v_k)\end{array}$}}}
%{\tau \in M_1(v_i)}
\frac{ 3 s_{\tau}((\mathbf{N}_{\tau} \cdot \mathcal{N}_{\tau}) \mathcal{N}_{\tau} - \mathbf{N}_{\tau}) \times (v_{k} - v_{j})}{2\mathcal{S}_{\tau}}+\eta(v_i-v_i^{in}).
\end{equation}

%\noindent \textbf{\textit{Remark.}}
%\begin{remark}
%The filtered face normal $\mathbf{N}_{\tau}$ and original area $s_{\tau}$ of triangle $\tau$ are fixed during the vertex updating procedure. The iterative updating normal $\mathcal{N}_{\tau}$ and area $\mathcal{S}_{\tau}$ need to be computed in each iteration according to updating vertices $(v_i, v_j, v_k)$ of $\tau.$ \rcomment{Not clear. You mean compute normal again after computing vertices, then iterate? Or write a whole procedure to clarify}
%\end{remark}

With the given gradient information \cref{eq:6} and the vertex positions of the initial noisy surface, many popular optimization techniques, such as Accelerated Gradient Descent and Quasi-Newton methods, can be used to solve our model \cref{eq:vertexUpdateModel}.
In this paper, we choose the BFGS algorithm \cite{Dennis1977Quasi}, which is one of the most commonly used methods for solving nonconstrained problems like
\cref{eq:vertexUpdateModel}.
In each iteration, BFGS algorithm uses only the energy and gradient evaluated at the current and previous iterations.

\Cref{fig:vertexUpdatingComparisons} demonstrates that our method \cref{eq:vertexUpdateModel} can greatly reduce foldovers compared to the method \cref{eq:Sun} proposed by Sun et al. \cite{Sun:07}.
%In this example, both of them use the ground truth face normals as input.
From the energy evolution curves in \cref{fig:vertexUpdatingComparisons}, we observe that, both methods are convergent and the iteration numbers of these two are close.
However, the results produced by \cite{Sun:07} suffer from severe foldovers (highlighted in red) and are inaccurate, while our method produces much better results without foldovers.
%see results in the second and third columns of Fig. \ref{fig:vertexUpdatingComparisons}.

\begin{figure*}[htb]
  \centering
  \includegraphics[width=1.0\linewidth]{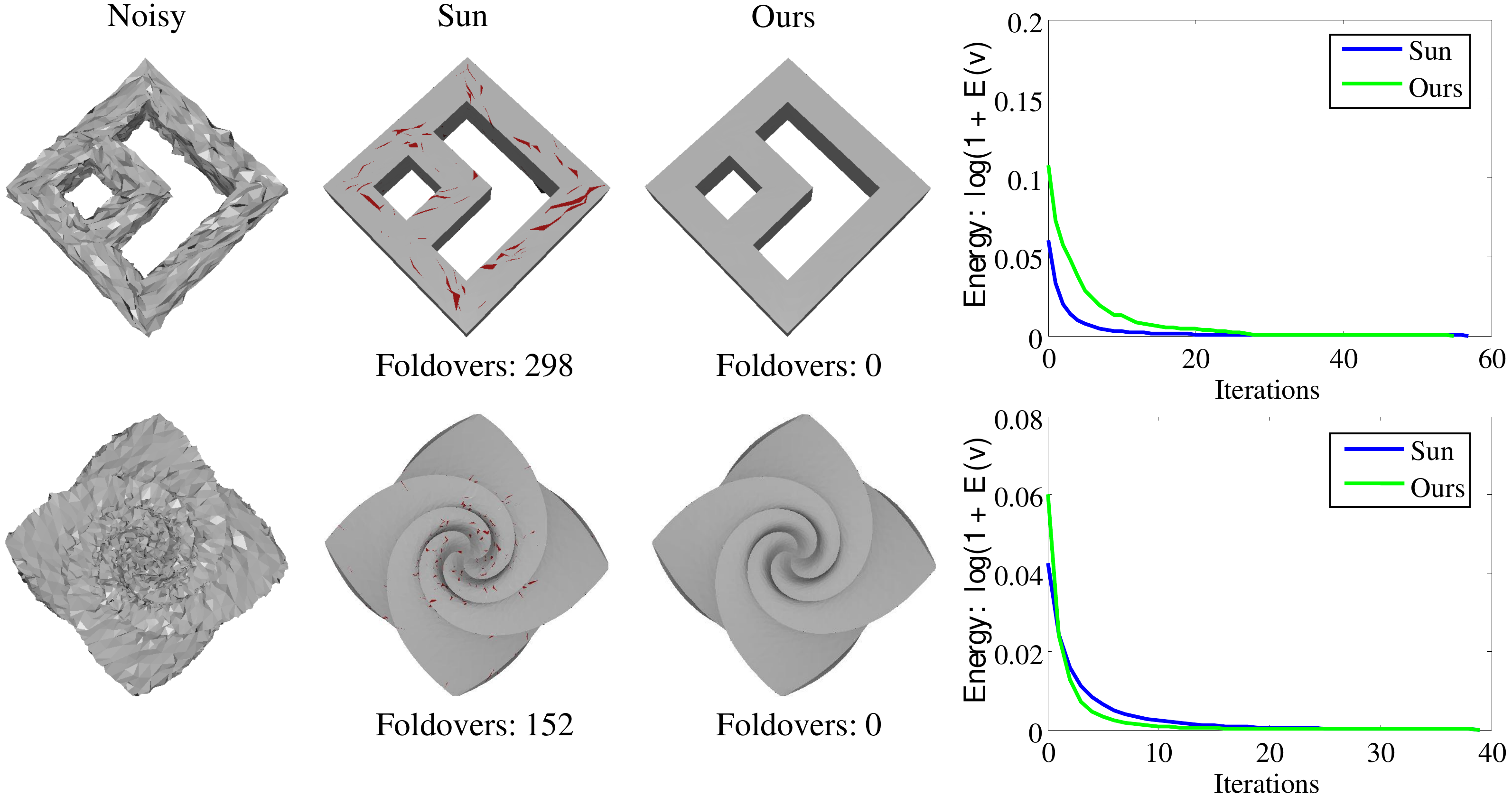}
  \caption{
  Comparisons of vertex updating methods (the method \cref{eq:Sun} proposed by Sun et al. \cite{Sun:07} and ours \cref{eq:vertexUpdateModel}).
  The first column shows the triangulated surfaces corrupted by Gaussian noise with standard deviation $\sigma = 0.4$ mean edge length and $\sigma=0.3$ mean edge length in random directions.
  The second and third columns are results produced by \cite{Sun:07} and ours, respectively.
  The foldovers are highlighted in red.
  The results produced by ours are without foldovers.
  The last column illustrates the energy evolution curves via iteration numbers.
  \label{fig:vertexUpdatingComparisons}}
\end{figure*}

\section{Numerical Experiments} \label{sec:results}
We verify the effectiveness of our two-stage denoising method on a variety of triangulated surfaces with either synthetic or raw noise.
The synthetic noise added in random directions is produced by a zero-mean Gaussian function with standard deviation $\sigma$ proportional to the mean edge length of the clean surface.
The clean surfaces tested in this section are listed in \cref{fig:modelList}.
%In addition, the synthetic noise corrupt the clean mesh in random directions.
\begin{figure*}
  \centering
  \includegraphics[width=1.0\linewidth]{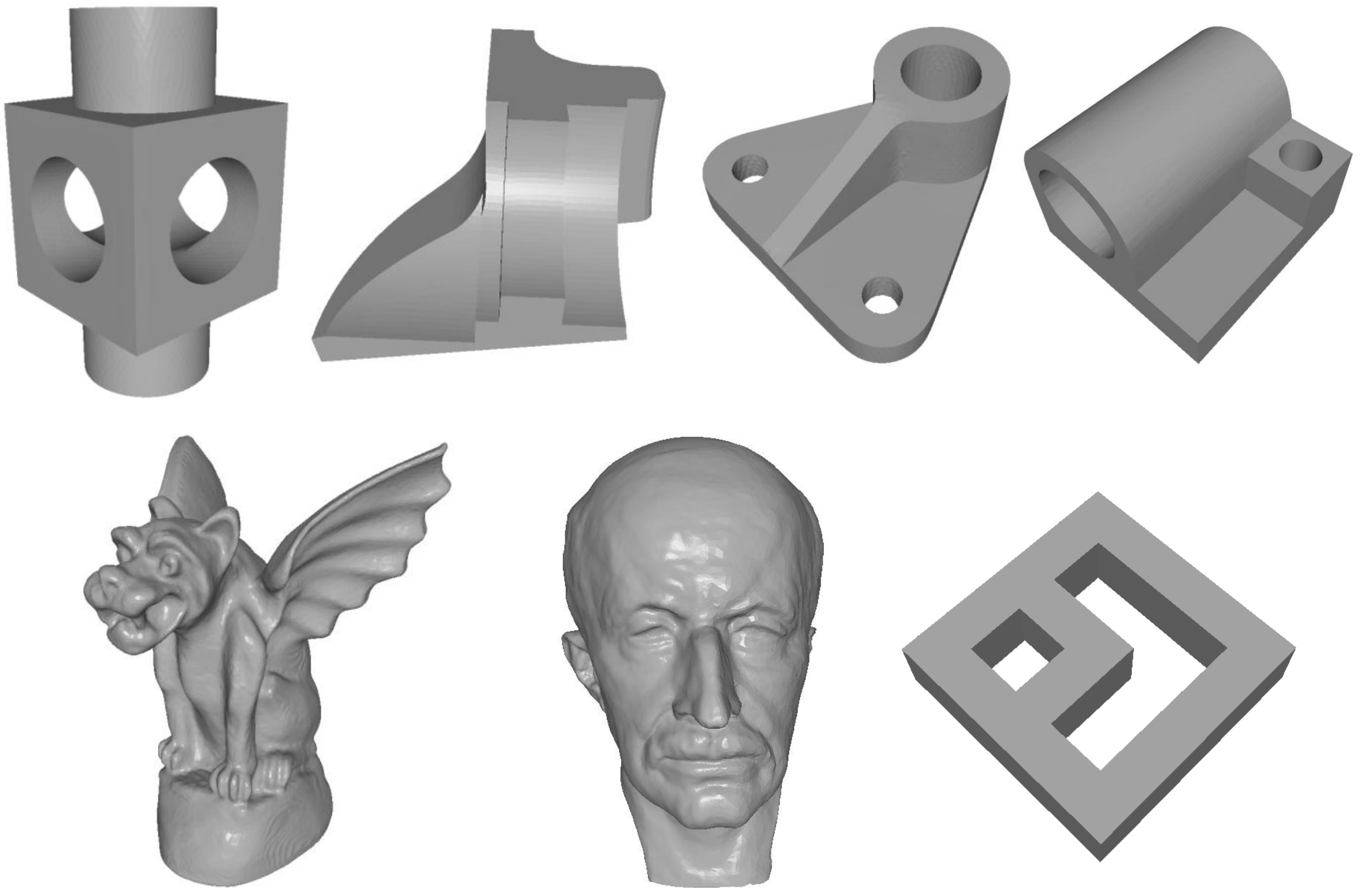}
  \caption{The clean surfaces tested in \cref{sec:results}.
  From top left to bottom right: Block, Fandisk, Part, Joint, Gargoyle, Max-Planck and Doubletorus.
%  (a) Block; (b)Fandisk; (c) Part;
%  (d) Joint; (e) Gargoyle; (f) Max-Planck; (g) Doubletorus.
  \label{fig:modelList}}
\end{figure*}

To verify the robustness of our denoising method to quality of surface triangles, we use two quantities as in \cite{liu2017geodesic}.
These quantities are defined as following:
\begin{equation*}
  D^{global} = \frac{\mathrm{min}_{\tau} \  \mathrm{area} \ \mathrm{of} \  \tau}{\mathrm{max}_{\tau} \  \mathrm{area} \  \mathrm{of} \ \tau},
\end{equation*}
\begin{equation*}
  D^{local} = \min_{\tau} \frac{\min_{e\prec\tau} \  \mathrm{length \ of \ } e }{\max_{e\prec\tau} \  \mathrm{length \ of \ } e }.
\end{equation*}
$D^{global}$ stands for the smallest largest triangle area ratio used for globally describing the distribution of triangles.
$D^{local}$ denotes the smallest one of ratios of shortest and longest edge lengths in triangles, which can be used to locally measure the quality of triangles.
The information of the clean surfaces are listed in \cref{tab:information}.
Although several surfaces including Gargoyle, Max-Planck and Embossment are not with very regular meshes as $D^{global}$ and $D^{local}$ indicated in \cref{tab:information},
our method can still effectively handle all these surfaces and produce satisfactory results.
%From Table \ref{tab:information} we can see that, several surfaces are very irregular, e.g., the Max-Planck surface. Our method can effectively handle all the surfaces and produce satisfatory results.

\begin{table*}[htb]
\centering
\caption{Information of surfaces tested in this paper
}
\begin{tabular}{|c|c|c|c|c|}
  \hline
  Surface    & \#vertices & \#triangles & $D^{global}$ & $D^{local}$ \\
  \hline
  Block     & 8771  & 17500  & 0.066312 &  0.369025 \\
  %\hline
  Fandisk    & 6475  & 12946  & 0.0202721 & 0.333103 \\
  Part       & 4261  & 8530   & 0.0596278 & 0.235325 \\
  Joint      & 5636  & 11276  & 0.000489636 & 0.0508141 \\
  Gargoyle   & 25002 & 50000  & 0.000194802 & 0.0814815 \\
  Max-Planck & 30942 & 61880  & 9.28547e-006 & 0.0135542 \\
  %Bunny     & 34835 & 69666  & 4.99319e-006 & 0.0632449 \\
  %\hline
  Rabbit    & 37394 & 73679   & 0.0805432 & 0.0991931 \\
  %\hline
  Angel     & 24566 & 72690  & 0.041708 & 0.0824533 \\
  %\hline
  Shell    & 58354 & 174031  & 0.00645787 & 0.106432 \\
  %\hline
  Embossment     & 65988  & 195095   & 0.00735766 & 0.0106189 \\
  %\hline
  Doubletorus    & 2686  & 5376   & 0.00439037 & 0.109461 \\
  %\hline
  %Vase & 14859 & 44601 & 0.0016335 & 0.158173 \\
  \hline
\end{tabular}
\label{tab:information}
\end{table*}

For fair comparisons, we have implemented all the algorithms tested in this paper using C++ and run all examples on a notebook with a Intel dual core 2.10 GHz processor and a 8GB RAM.
All the surfaces are rendered in flat-shading model to show faceting effect.
Our algorithm is compared qualitatively and quantitatively to state-of-art methods, respectively.
We also discuss our algorithm from various aspects, including influences of parameters and algorithm convergence.

\subsection{Qualitative Comparisons}
In this subsection, we compare our surface denoising method w-HO with other methods including TV normal filtering method \cite{Zhang:15},  $\ell_{0}$ minimization \cite{He13} and bilateral weighting Laplacian optimization \cite{Zheng:11}, abbreviated as TV, $\ell_{0}$ and bw-Laplacian respectively.
For all these methods, we carefully tuned the parameters to get the visually best denoising results.

In \cref{fig:cadComparison}, we compare the results for surfaces containing both sharp features (including sharp edges and corners) and smooth regions (including smoothly curved regions and flat regions).
As we can see, bw-Laplacian keeps smooth regions well but blurs sharp features, while our w-HO method, TV and $\ell_0$ preserve most sharp features well.
Furthermore, TV and $\ell_0$ both suffer from staircase effects in smoothly curved regions indicated in \cref{fig:cadComparison}(c) and (d), and this phenomenon is extremely serious for $\ell_0$ which produces false edges in the first and last row of (d) of \cref{fig:cadComparison}.
However, our w-HO method does not produce the staircase effect while preserving sharp features well.
%The reasons may be as follows.
As we know, both sharp features and noise belong to high frequency information.
The bw-Laplacian cannot distinguish them strictly, especially for small scale features.
Thus, it may treat some features as noise and blur them.
In addition, as stated in compressed sensing, both $\ell_0$  norm and $\ell_1$ norm have sparse property, which can be used for preserving sharp features.
However, as $\ell_0$ and TV use low order information of surfaces, they tend to produce staircase effects in smooth regions, especially for $\ell_0$ for its high sparsity requirement.
Consequently, the compared three methods can either deal with smooth regions or sharp features well.
In contrast, our w-HO method can suppress the staircase effects in smooth regions and simultaneously preserve sharp features.
In all, for CAD-like surfaces, visual comparisons in \cref{fig:cadComparison} show that our w-HO method
is noticeably better than all the other three methods in terms of smooth regions and sharp features recovery.

%The results generated by TV method \cite{Zhang:15} show some staircase effect and the ones produced by $L_0$ minimization \cite{He13} suffer from severe staircasing.
%From (f) in Fig. \ref{fig:cadComparison}, we find that the staircase effect in smooth regions is suppressed but sharp features are not recovered that accurately of the tested meshes.
%In contrast, our method recovers both sharp features and smooth regions well, and doest not generate staircase effect.
%Consequently, visual comparisons in Fig.\ref{fig:cadComparison} show that Ours has best denoising results among all the compared methods \cite{Zhang:15,He13,Zheng:11}.
%Then, we quantitatively compare our method to other methods by three types numerical errors in Table \ref{tab:HOTVAndTVComparing}.
%As can be seen, for CAD meshes, the denoising results by our method have least errors in Mean square angular error(MSAE).
%And for $L_2$ vertex-based mesh-to-mesh and $L_{\infty}$ vertex-based mesh-to-mesh error, our method still works well.
%In summary, for CAD meshes consist of sharp features and smooth regions, visual and quantitative comparisons demonstrate that our method outperforms typical existing methods.

\begin{figure}[tbhp]
  \centering
  \subfloat[Noisy]{\label{fig:cadComparison-a}\includegraphics[width=0.2\textwidth]{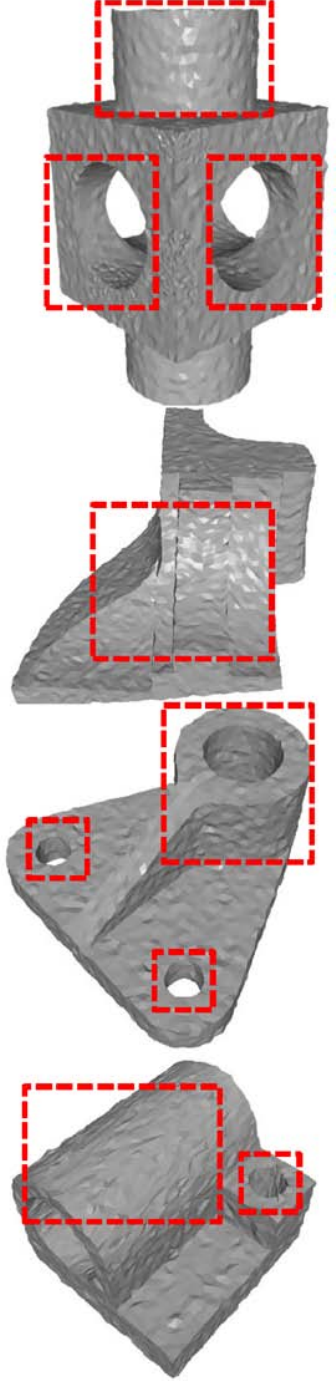}}
  \subfloat[w-HO]{\label{fig:cadComparison-b}\includegraphics[width=0.2\textwidth]{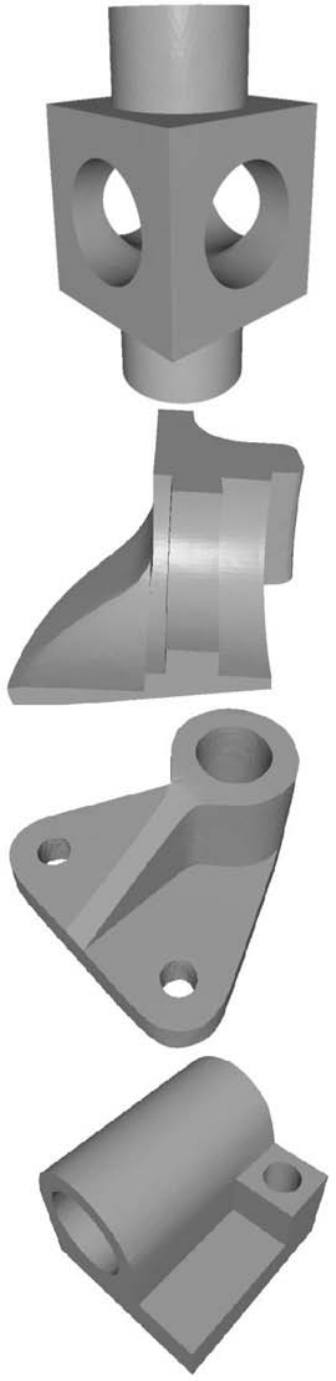}}
  \subfloat[TV]{\label{fig:cadComparison-c}\includegraphics[width=0.2\textwidth]{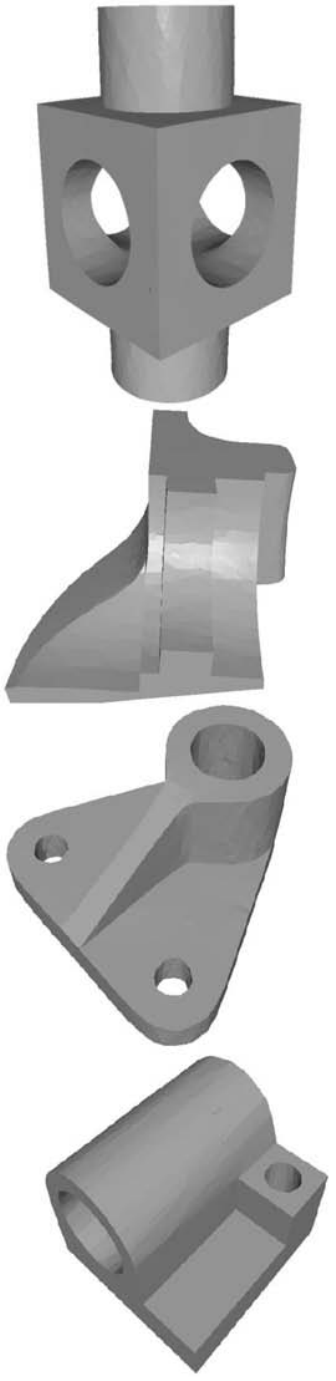}}
  \subfloat[$\ell_0$]{\label{fig:cadComparison-d}\includegraphics[width=0.2\textwidth]{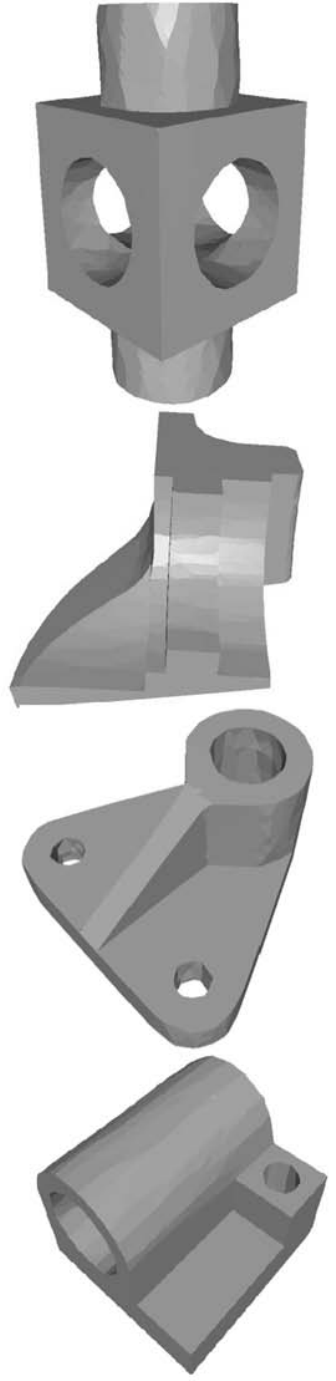}}
  \subfloat[bw-Laplacian]{\label{fig:cadComparison-e}\includegraphics[width=0.2\textwidth]{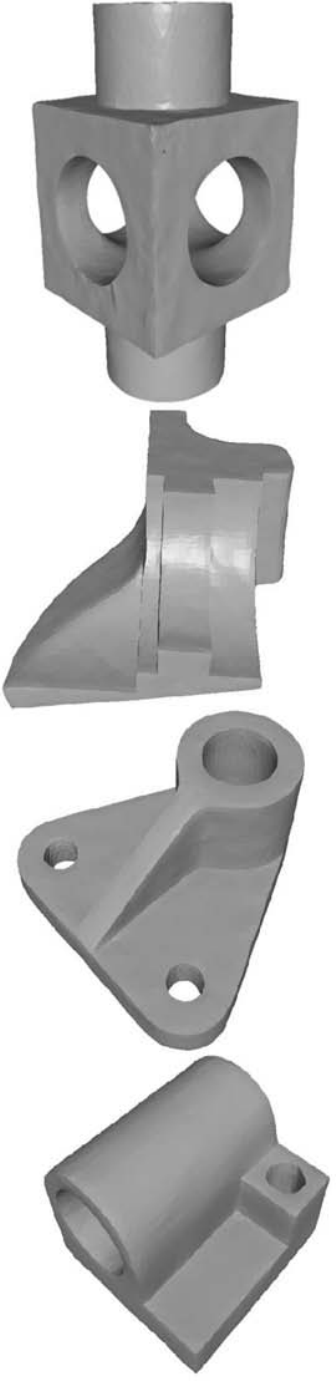}}
  \caption{Denoising results of Block, Fandisk, Part and Joint (corrupted by Gaussian noise, standard deviation = 0.15 mean edge length).
  From left to right: input noisy surfaces, denoising results produced by our proposed w-HO method, TV method \cite{Zhang:15}, $\ell_{0}$ minimization \cite{He13} and bilateral weighting Laplacian method \cite{Zheng:11}, respectively.}
  \label{fig:cadComparison}
\end{figure}

%\begin{figure*}
%  \centering
%  \includegraphics[width=1.0\linewidth]{figs/cadComparison}
%  \caption{Denoising results of Block, Fandisk, Part and Joint (corrupted by Gaussian noise, standard deviation = 0.15 mean edge length).
%  From left to right: input noisy surfaces, denoising results produced by our proposed w-HO method, TV method \cite{Zhang:15}, $\ell_{0}$ minimization \cite{He13} and bilateral weighting Laplacian method \cite{Zheng:11}, respectively.
%  \label{fig:cadComparison}}
%\end{figure*}

\Cref{fig:noncadComparison} shows results of surfaces with fine features.
As can be seen, TV and $\ell_0$ tend to flatten some details, and $\ell_0$ performs even worse.
%(see the magnified views of Fig.~\ref{fig:noncadComparison}).
%Especially, $\ell_0$ minimization \cite{He13} suffers from severe staircase effects.
Our w-HO method and bw-Laplacian can both generate visually better denoising results.
However, from numerical metrics (which will be introduced in \cref{sec:6.2}), we observe that errors of our method are always lower than those of bw-Laplacian.
This demonstrates that our method is better than bw-laplacian.
In general, for non-CAD surfaces, our w-HO method can also yield satisfactory results containing more details than other methods.

\begin{figure}[tbhp]
  \centering
  \subfloat[Noisy]{\label{fig:noncadComparison-a}\includegraphics[width=0.2\textwidth]{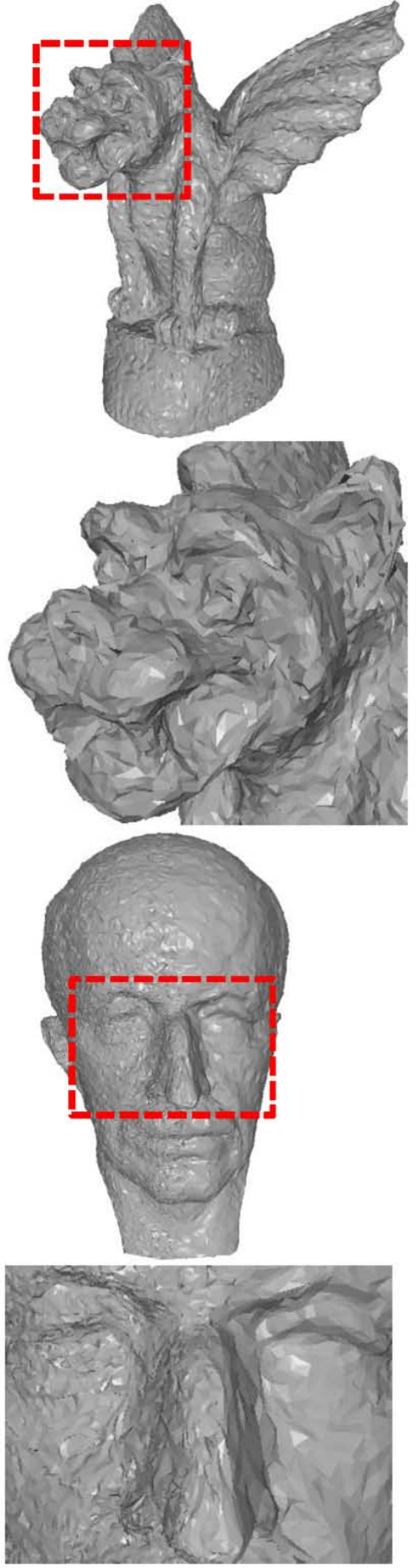}}
  \subfloat[w-HO]{\label{fig:noncadComparison-b}\includegraphics[width=0.2\textwidth]{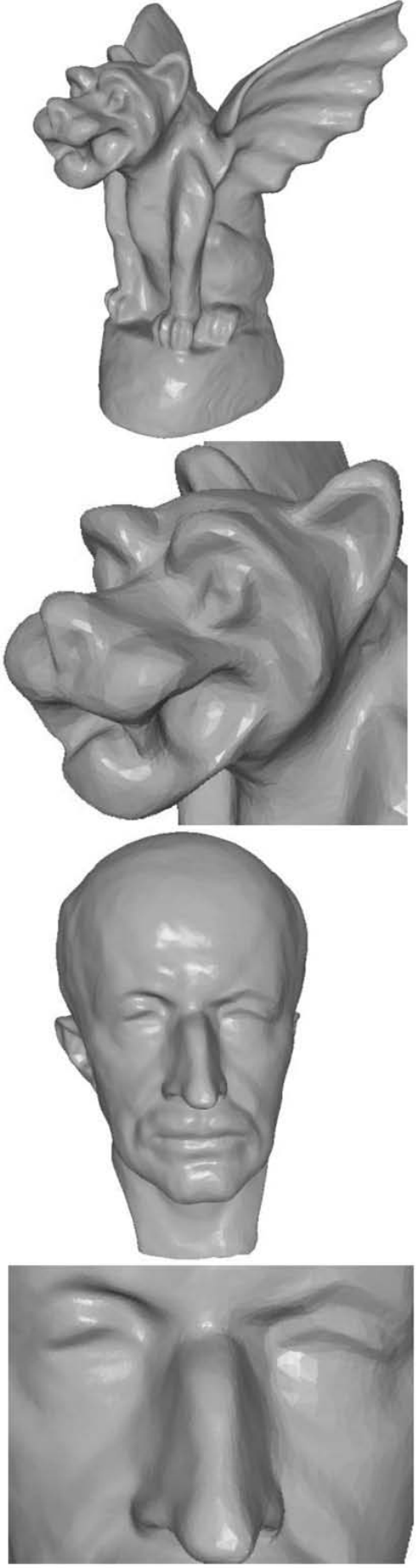}}
  \subfloat[TV]{\label{fig:noncadComparison-c}\includegraphics[width=0.2\textwidth]{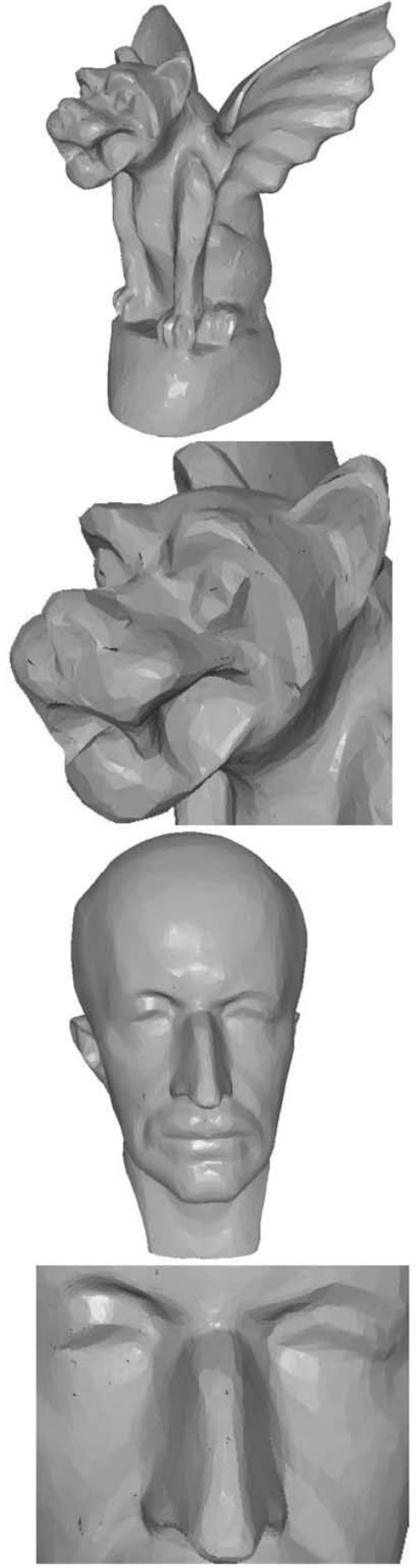}}
  \subfloat[$\ell_0$]{\label{fig:noncadComparison-d}\includegraphics[width=0.2\textwidth]{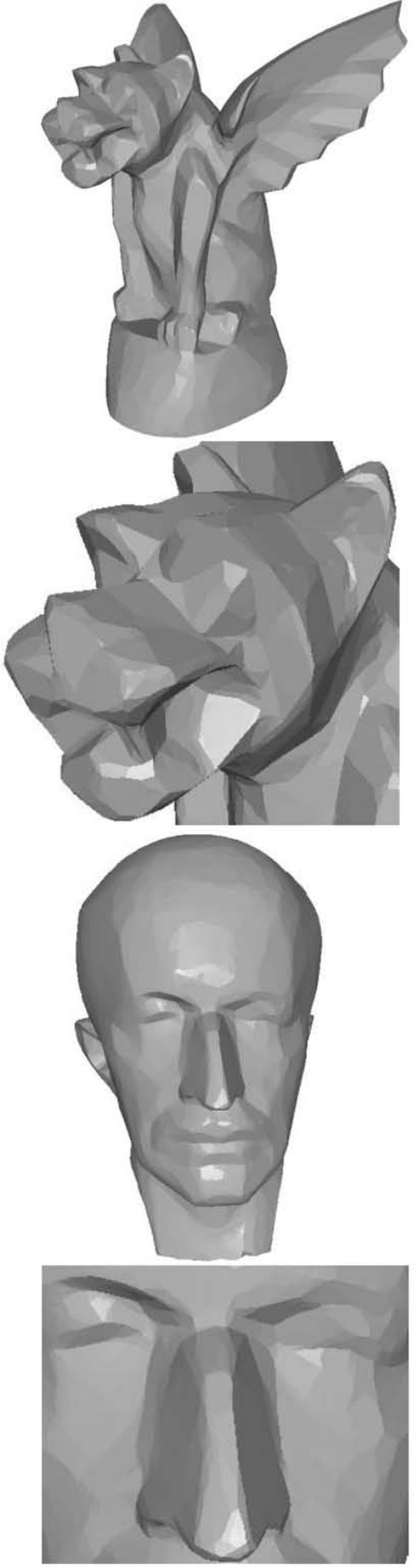}}
  \subfloat[bw-Laplacian]{\label{fig:noncadComparison-e}\includegraphics[width=0.2\textwidth]{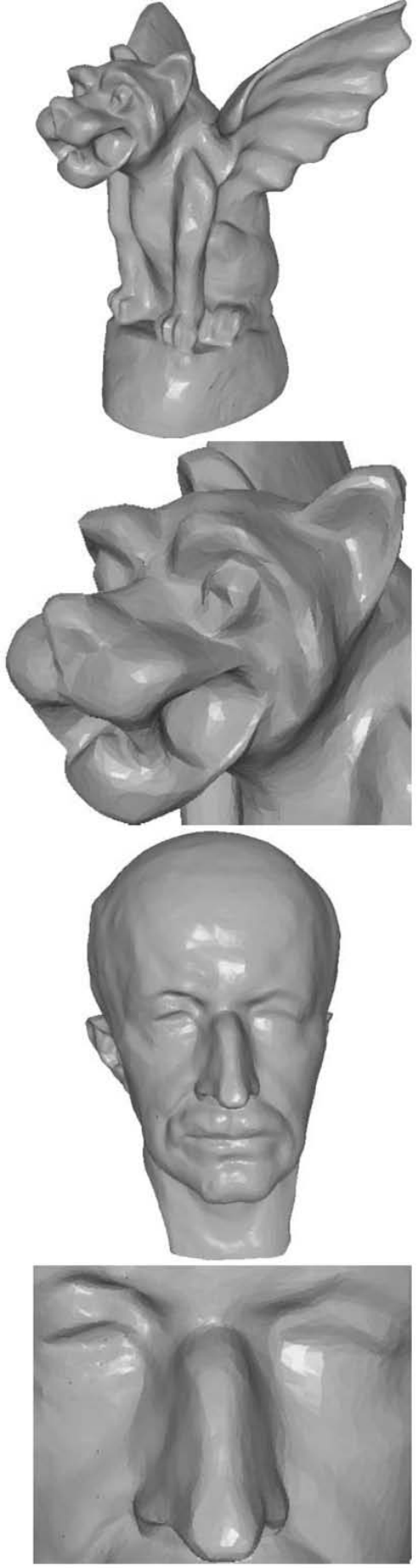}}
  \caption{Denoising results of Gargoyle (corrupted by Gaussian noise, standard deviation = 0.25 mean edge length) and Max-Planck (corrupted by Gaussian noise, standard deviation = 0.2 mean edge length).
  From left to right: input noisy surfaces, denoising results produced by our proposed w-HO method, TV method \cite{Zhang:15}, $\ell_{0}$ minimization \cite{He13} and bilateral weighting Laplacian method \cite{Zheng:11}, respectively.
  The second and fourth rows show magnified views of Gargoyle and Max-Planck.}
  \label{fig:noncadComparison}
\end{figure}

%\begin{figure*}
%  \centering
%  \includegraphics[width=1.0\linewidth]{figs/noncadComparison}
%  \caption{Denoising results of Gargoyle (corrupted by Gaussian noise, standard deviation = 0.25 mean edge length) and Max-Planck (corrupted by Gaussian noise, standard deviation = 0.2 mean edge length).
%  From left to right: input noisy surfaces, denoising results produced by our proposed w-HO method, TV method \cite{Zhang:15}, $\ell_{0}$ minimization \cite{He13} and bilateral weighting Laplacian method \cite{Zheng:11}, respectively.
%  The second and fourth rows show magnified views of Gargoyle and Max-Planck.
%  \label{fig:noncadComparison}}
%\end{figure*}

To further demonstrate the validity of our w-HO method, we test it on real scanned surfaces; see \cref{fig:scandata}.
We can see that, our method can yield very good denoising results preserving most features well.

\begin{figure*}
  \centering
  \includegraphics[width=1.0\linewidth]{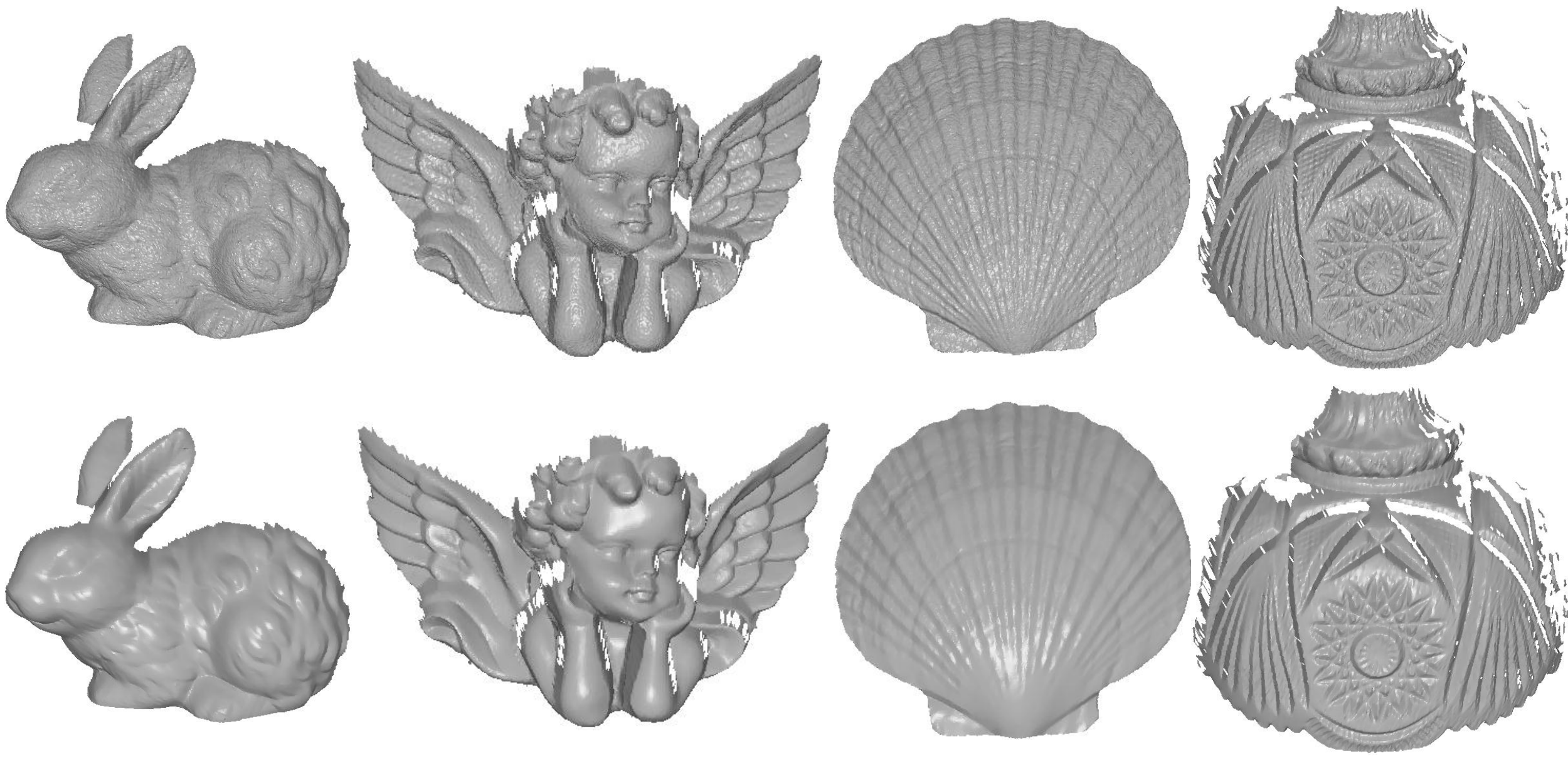}
  \caption{Our denoising results for four real scanned surfaces. From left to right: Rabbit, Angel, Shell and Embossment.
  \label{fig:scandata}}
\end{figure*}

\subsection{Quantitative Comparisons} \label{sec:6.2}
\begin{table}
\centering
\caption{Quantitative evaluation results of \cref{fig:cadComparison,fig:noncadComparison} for our proposed w-HO method, TV method \cite{Zhang:15}, $\ell_{0}$ minimization \cite{He13} and bilateral weighting Laplacian method \cite{Zheng:11}.
$\sigma$ is the standard deviation of the Gaussian noise added to the clean surface.} \label{tab:HOTVAndTVComparing}
\begin{tabular}{|c|c|c|c|c|c|}
  \hline
  \multicolumn{1}{|c|}{\multirow{2}{*}{Models}} & \multicolumn{1}{|c|}{\multirow{2}{*}{$\sigma$}} &  \multicolumn{1}{|c|}{\multirow{2}{*}{Methods}} & \multicolumn{1}{|c|}{$\mathrm{MSAE}$}   &
  \multicolumn{1}{|c|}{$\mathcal{E}_{v,2}$} &
  \multicolumn{1}{|c|}{CPU costs}
  \\

  \multicolumn{1}{|c|}{}                  &  \multicolumn{1}{|c|}{}          & \multicolumn{1}{|c|}{}     & \multicolumn{1}{|c|}{$(\times 10^{-3})$} & \multicolumn{1}{|c|}{$(\times 10^{-3})$} &
  \multicolumn{1}{|c|}{(in seconds)}
  \\
  \hline

    \multicolumn{1}{|c|}{\multirow{4}{*}{Block}}  & \multicolumn{1}{|c|}{\multirow{4}{*}{0.15}} &
    \multicolumn{1}{|c|}{w-HO} & $\mathbf{2.40}$  & $\mathbf{0.79}$  & 7.35          \\
     \multicolumn{1}{|c|}{}  &                    & TV         & $3.61$  & $0.98$ &2.13             \\
    \multicolumn{1}{|c|}{}  &                     & $\ell_0$      & $4.31$  & $1.81$ &16.57            \\
    \multicolumn{1}{|c|}{}  &                     & bw-Laplacian  & $5.66$  & $1.01$ &1.45          \\
    \cline{1-6}

  \multicolumn{1}{|c|}{\multirow{4}{*}{Fandisk}}  & \multicolumn{1}{|c|}{\multirow{4}{*}{0.15}}
  & \multicolumn{1}{|c|}{w-HO}                                  & $\mathbf{1.48}$  & $0.88$ &1.68            \\
   \multicolumn{1}{|c|}{}  &                      & TV          & $1.53$  & $\mathbf{0.86}$ &0.84         \\
   \multicolumn{1}{|c|}{}  &                      & $\ell_0$       & $3.85$  & $1.12$          &7.53 \\
   \multicolumn{1}{|c|}{}  &                      & bw-Laplacian   & $2.50$  & $1.01$          &0.76   \\
  \cline{1-6}

  \multicolumn{1}{|c|}{\multirow{4}{*}{Part}}  & \multicolumn{1}{|c|}{\multirow{4}{*}{0.15}}
  & \multicolumn{1}{|c|}{w-HO}                                  & $\mathbf{1.29}$  & $\mathbf{0.94}$ & 1.84             \\
   \multicolumn{1}{|c|}{}  &                      & TV          & $2.51$  & $1.22$ & 0.67         \\
   \multicolumn{1}{|c|}{}  &                      & $\ell_0$       & $8.1$   & $2.32$ & 8.01         \\
   \multicolumn{1}{|c|}{}  &                      & bw-Laplacian   & $4.22$  & $1.23$ & 0.53             \\
  \cline{1-6}

  \multicolumn{1}{|c|}{\multirow{4}{*}{Joint}}  & \multicolumn{1}{|c|}{\multirow{4}{*}{0.15}}
  & \multicolumn{1}{|c|}{w-HO}                                  & $\mathbf{2.88}$  &$\mathbf{0.76}$ &2.80             \\
   \multicolumn{1}{|c|}{}  &                      & TV          & $4.21$  & $1.20$ & 1.98        \\
   \multicolumn{1}{|c|}{}  &                      & $\ell_0$       & $12.7$  & $2.37$ & 12.84         \\
   \multicolumn{1}{|c|}{}  &                      & bw-Laplacian   & $6.77$  & $2.01$ & 0.81            \\
  \cline{1-6}

  \multicolumn{1}{|c|}{\multirow{4}{*}{Bunny}}  & \multicolumn{1}{|c|}{\multirow{4}{*}{0.2}} & \multicolumn{1}{|c|}{w-HO} & $\mathbf{12.9}$  & $0.92$  &18.78        \\
   \multicolumn{1}{|c|}{}  &                      & TV      & $16.1$  & $0.89$ &8.57          \\
   \multicolumn{1}{|c|}{}  &                      & $\ell_0$   & $27.5$  & $2.16$ &60.72          \\
   \multicolumn{1}{|c|}{}  &                      & bw-Laplacian   & $13.4$  & $\mathbf{0.88}$ &7.89             \\
  \cline{1-6}

  \multicolumn{1}{|c|}{\multirow{4}{*}{Gargoyle}}  & \multicolumn{1}{|c|}{\multirow{4}{*}{0.25}} & \multicolumn{1}{|c|}{w-HO} & $\mathbf{13.5}$  & $0.77$ &23.32        \\
  \multicolumn{1}{|c|}{}  &                      & TV   & $17.3$  & $0.83$ &12.34           \\
   \multicolumn{1}{|c|}{}  &                      & $\ell_0$   & $31.0$  & $1.79$ &57.36            \\
   \multicolumn{1}{|c|}{}  &                      & bw-Laplacian   & $16.1$  & $\mathbf{0.75}$ &5.12              \\
  \cline{1-6}

    \multicolumn{1}{|c|}{\multirow{4}{*}{Max-Planck}}  & \multicolumn{1}{|c|}{\multirow{4}{*}{0.2}} & \multicolumn{1}{|c|}{w-HO} & $\mathbf{10.8}$  & $\mathbf{0.85}$ &31.81        \\
  \multicolumn{1}{|c|}{}  &                       & TV          & $16.6$  & $1.11$ &13.41           \\
   \multicolumn{1}{|c|}{}  &                      & $\ell_0$    & $33.5$  & $1.86$ &74.05            \\
   \multicolumn{1}{|c|}{}  &                      & bw-Laplacian   & $12.1$  & $0.91$ &9.41              \\
  \cline{1-6}

\end{tabular}
\end{table}

From the above comparisons, we find that our w-HO method generates visually better results than those compared methods.
In this subsection, we further compare them quantitatively.

We use two error metrics \cite{Sun:07,Sun:08,Zheng:11} to measure the deviation of the denoised surface from the clean one, which are defined as followed:

%We use three numerical errors to measure the deviation of the denoising result from the clean mesh.
%These three error metrics \cite{Sun:07,Sun:08,Zheng:11} are defined as followed to measure the angular, mean and maximal error, respectively.
\begin{itemize}
  \item Mean square angular error (MSAE):
  \begin{equation*}
    \mathrm{MSAE} = \mathrm{average}({\angle(\mathbf{N}^{'},\mathbf{N})}),
  \end{equation*}
  where $\angle(\mathbf{N}^{'},\mathbf{N})$ is the square angle between the normal of the denoising result and the clean surface,  $\mathrm{average}({\angle(\mathbf{N}^{'},\mathbf{N})})$ is the square angle averaged over all faces.

  \item $L_2$ vertex-based surface-to-surface error:
  \begin{equation*}
    \mathcal{E}_{v,2} = \sqrt{\frac{1}{3\sum\limits_{\tau}s_\tau} \sum\limits^{\mathrm{V}-1}_{i=0}(\sum\limits_{M_1(v_i)}s_\tau) \mathrm{dist}(v^{,}_i,M)^2},
  \end{equation*}
  where $\mathrm{dist}(v^{,}_i,M)$ is the distance between the updated vertex $v^{'}_{i}$ and a triangle of the clean surface $M$ which is closest to $v^{'}_{i}$.
  %This metric is to measure the positional error between the clean surface and the denoised result.

%  \item $L_{\infty}$ vertex-based mesh-to-mesh error (vertex-based Hausdorff distance):
%  \begin{equation*}
%    E_{v,\infty} = \mathrm{max}_{i}\mathrm{dist}(v^{,}_{i},M).
%  \end{equation*}
\end{itemize}

Then, we compare our w-HO method to other three methods using the above two error metrics for the examples shown in \cref{fig:cadComparison,fig:noncadComparison}.
The evaluation results are listed in \cref{tab:HOTVAndTVComparing}.
As can be seen, our w-HO method outperforms the other methods in the sense that angular errors (MSAE) from w-HO
are significantly smaller than all the other methods, especially for CAD-like surfaces.
%MSAEs' from w-HO are significantly smaller than all the other methods, especially for CAD-like surfaces.
It is also observed that, the results of w-HO have the least $L_2$ vertex-based errors ($\mathcal{E}_{v,2}$) in most cases.
This demonstrates that the results produced by w-HO are more faithful to the ground truth surfaces.

The CPU costs of all the tested methods are recorded in the last column of \cref{tab:HOTVAndTVComparing}.
For our w-HO method, the most time-consuming part is solving the $\textbf{N}$-sub problem.
As mentioned in \cref{sec:alm}, due to the error forgetting property \cite{Yin2013Error} of our ALM algorithm, we use a fast approximate strategy to solve this subproblem.
As can be seen, bilateral weighting Laplacian method \cite{Zheng:11} is the fastest method, while $\ell_0$ minimization \cite{He13} is the slowest.
Although our w-HO method is a little more computationally intensive than TV method \cite{Zhang:15}, the CPU cost is still acceptable.
In the future, we will investigate how to accelerate our w-HO method.

%\rcomment{It would be better if we can also list the computation time comparisons if the time consumption of our method is not too bad.}
%In addition, for meshes with fine fetures in Fig. \ref{fig:noncadComparison}, Laplacian \cite{Zheng:11} sometimes
%can produce competatvie results to results of Ours.
%works better in $E_{v,2}$ than Ours.
%However, for this metric, the values of the results produced by Ours are close to the best ones generated by Laplacian \cite{Zheng:11}.

\subsection{Influence of Parameters}
To our knowledge, most triangulated surface denoising methods have parameters, which need to be manually tuned.
\cref{alg:normalFiltering} also has two parameters, i.e., $\alpha$ and $r_\mathbf{p}$.
These two parameters need to be tuned for producing prominent results.
The first parameter is used to balance the fidelity and regularization term of the normal filtering model \cref{InitHOTVmodel}.
The second one is introduced by the augmented Lagrangian method.

$\alpha$ is used to control the degree of denoising and smoothness of the result surface.
\cref{fig:parameterAlpha} illustrates results of different $\alpha$ with fixed $r_\mathbf{p}$.
As can be seen, if $\alpha$ is too large, noise cannot be effectively removed indicated in \cref{fig:parameterAlpha}(b); and if $\alpha$ is too small, surfaces will be over-smoothed and fine features will be lost illustrated in \cref{fig:parameterAlpha}(e).
For each noisy surface, there exist a range of $\alpha$ for  \cref{alg:normalFiltering} producing visually well denoising results; see \cref{fig:parameterAlpha}(c) and (d).

\begin{figure}[tbhp]
  \centering
  \subfloat[Noisy]{\label{fig:parameterAlpha-a}\includegraphics[width=0.2\textwidth]{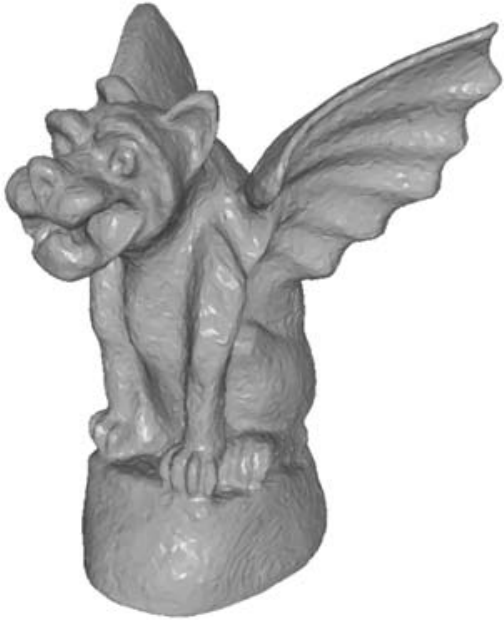}}
  \subfloat[$\alpha=3000$]{\label{fig:parameterAlpha-b}\includegraphics[width=0.2\textwidth]{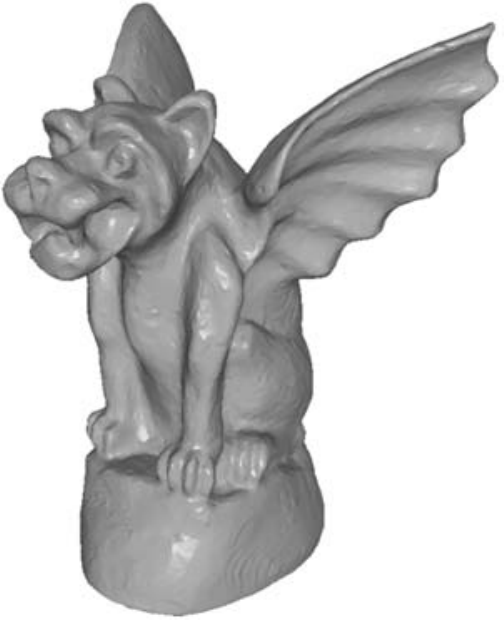}}
  \subfloat[$\alpha=1500$]{\label{fig:parameterAlpha-c}\includegraphics[width=0.2\textwidth]{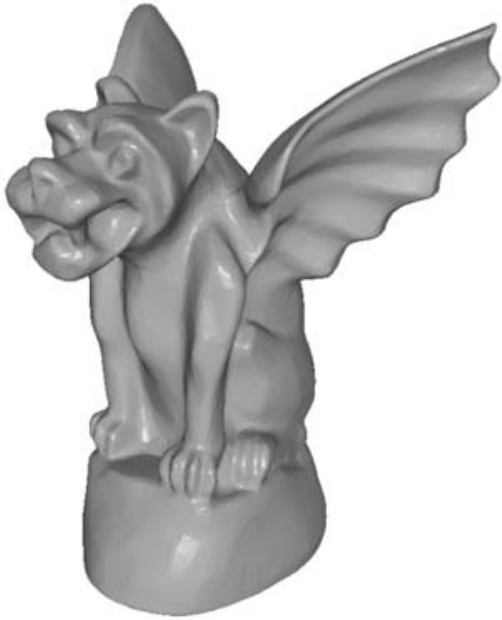}}
  \subfloat[$\alpha=600$]{\label{fig:parameterAlpha-d}\includegraphics[width=0.2\textwidth]{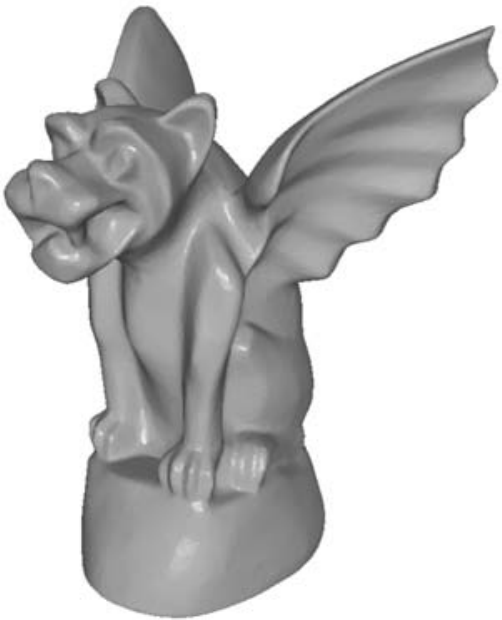}}
  \subfloat[$\alpha=200$]{\label{fig:parameterAlpha-e}\includegraphics[width=0.2\textwidth]{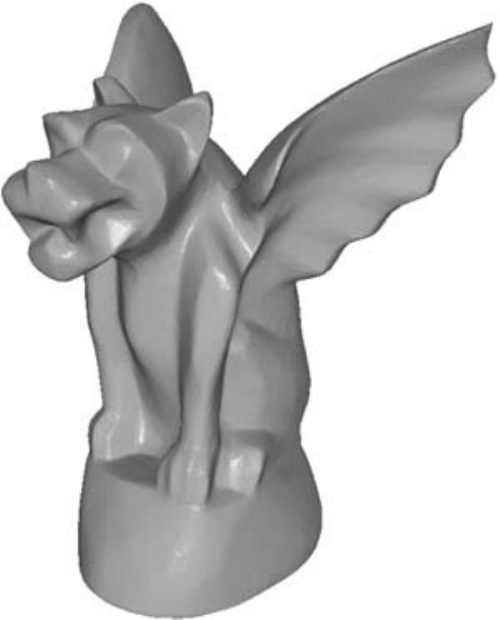}}
  \caption{Denoising results for $\alpha$ with fixed $r_\mathbf{p}$. From left to right: input noisy surface (corrupted by Gaussian noise, standard deviation = 0.1 mean edge length) and results with different $\alpha$.}
  \label{fig:parameterAlpha}
\end{figure}

%\begin{figure*}
%  \centering
%  \includegraphics[width=1.0\linewidth]{figs/parameterAlpha}
%  \caption{Denoising results for $\alpha$ with fixed $r_\mathbf{p}$. From left to right: input noisy surface (corrupted by Gaussian noise, standard deviation = 0.1 mean edge length) and denoising results with different $\alpha$.
%  \label{fig:parameterAlpha}}
%\end{figure*}

$r_\mathbf{p}$ also has influence on denoising results. \cref{fig:parameterR} shows results of different $r_\mathbf{p}$ with fixed $\alpha$.
As we can see, too small $r_\mathbf{p}$ will left some noise on the surface indicated in \cref{fig:parameterR}(b), and too large $r_\mathbf{p}$ should over-smooth the result illustrated in \cref{fig:parameterR}(e).
Again, for each noisy surface, there exist a range of $r_\mathbf{p}$ for our algorithm producing visually well results as shown in \cref{fig:parameterR}(c) and (d).

\begin{figure}[tbhp]
  \centering
  \subfloat[Noisy]{\label{fig:parameterR-a}\includegraphics[width=0.2\textwidth]{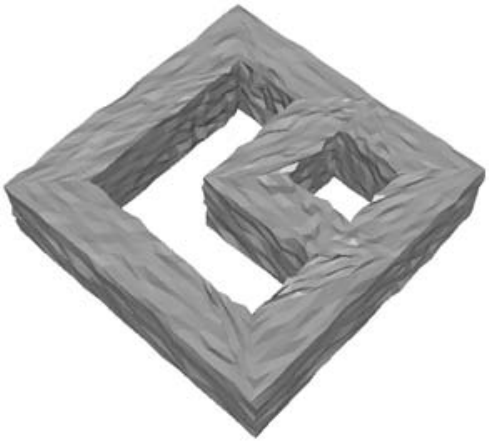}}
  \subfloat[$r_\mathbf{p}=0.1$]{\label{fig:parameterR-b}\includegraphics[width=0.2\textwidth]{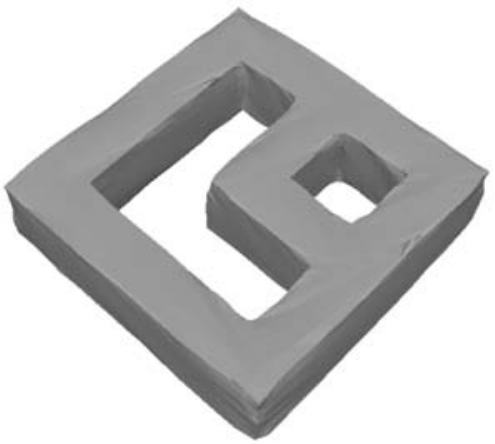}}
  \subfloat[$r_\mathbf{p}=1$]{\label{fig:parameterR-c}\includegraphics[width=0.2\textwidth]{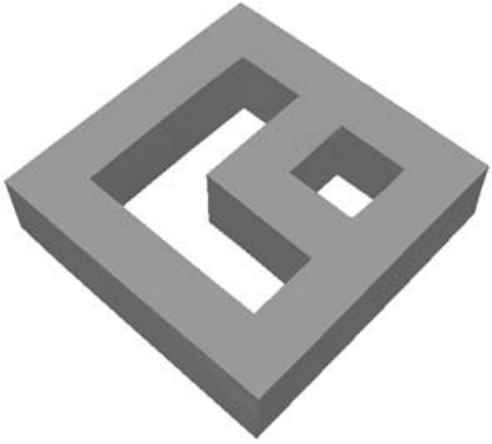}}
  \subfloat[$r_\mathbf{p}=5$]{\label{fig:parameterR-d}\includegraphics[width=0.2\textwidth]{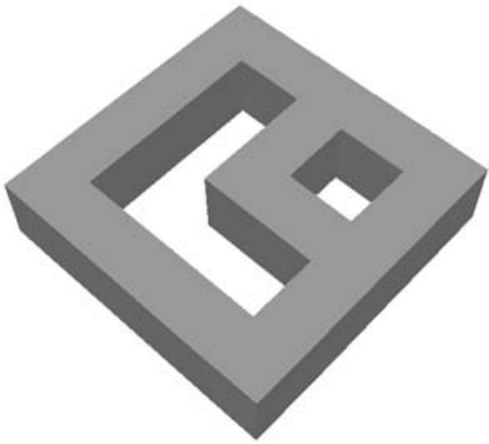}}
  \subfloat[$r_\mathbf{p}=100$]{\label{fig:parameterR-e}\includegraphics[width=0.2\textwidth]{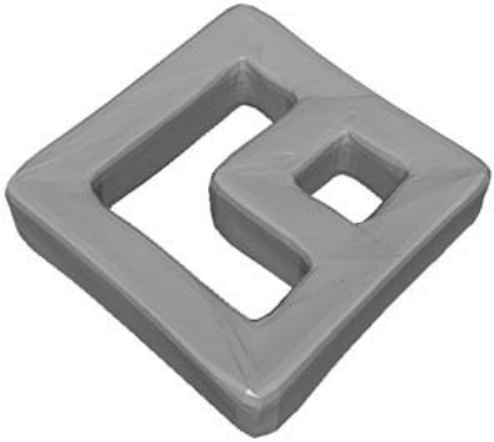}}
  \caption{Denoising results for $r_\mathbf{p}$ with fixed $\alpha$. From left to right: input noisy surface (corrupted by Gaussian noise, standard deviation = 0.15 mean edge length) and results with different $r_\mathbf{p}$.}
  \label{fig:parameterR}
\end{figure}

%\begin{figure*}
%  \centering
%  \includegraphics[width=1.0\linewidth]{figs/parameterR}
%  \caption{Denoising results for $r_\mathbf{p}$ with fixed $\alpha$. From left to right: input noisy surface (corrupted by Gaussian noise, standard deviation = 0.15 mean edge length) and denoising results with different $r_\mathbf{p}$.
%  \label{fig:parameterR}}
%\end{figure*}

\subsection{Algorithm Convergence and Effect of Dynamic Weights}
Due to nonlinear and nonconvex constraints of the proposed high order normal filtering model \cref{InitHOTVmodel},
it is a challenge to have the convergence analysis of \cref{alg:normalFiltering}.
However, we can verify the convergence using numerical experiments.
From the energy evolution in \cref{fig:energyCurve}, we observe that, the energy always decrease in each iteration.
This verifies the numerical convergence of \cref{alg:normalFiltering}.

\begin{figure}
 \captionsetup[subfigure]{justification=centering}
  \centering
  \begin{tabular}{c@{\hspace{3mm}}c@{\hspace{3mm}}}
  \subfloat[][]
  {\label{fig:curve-a} \includegraphics[width=0.42\textwidth]{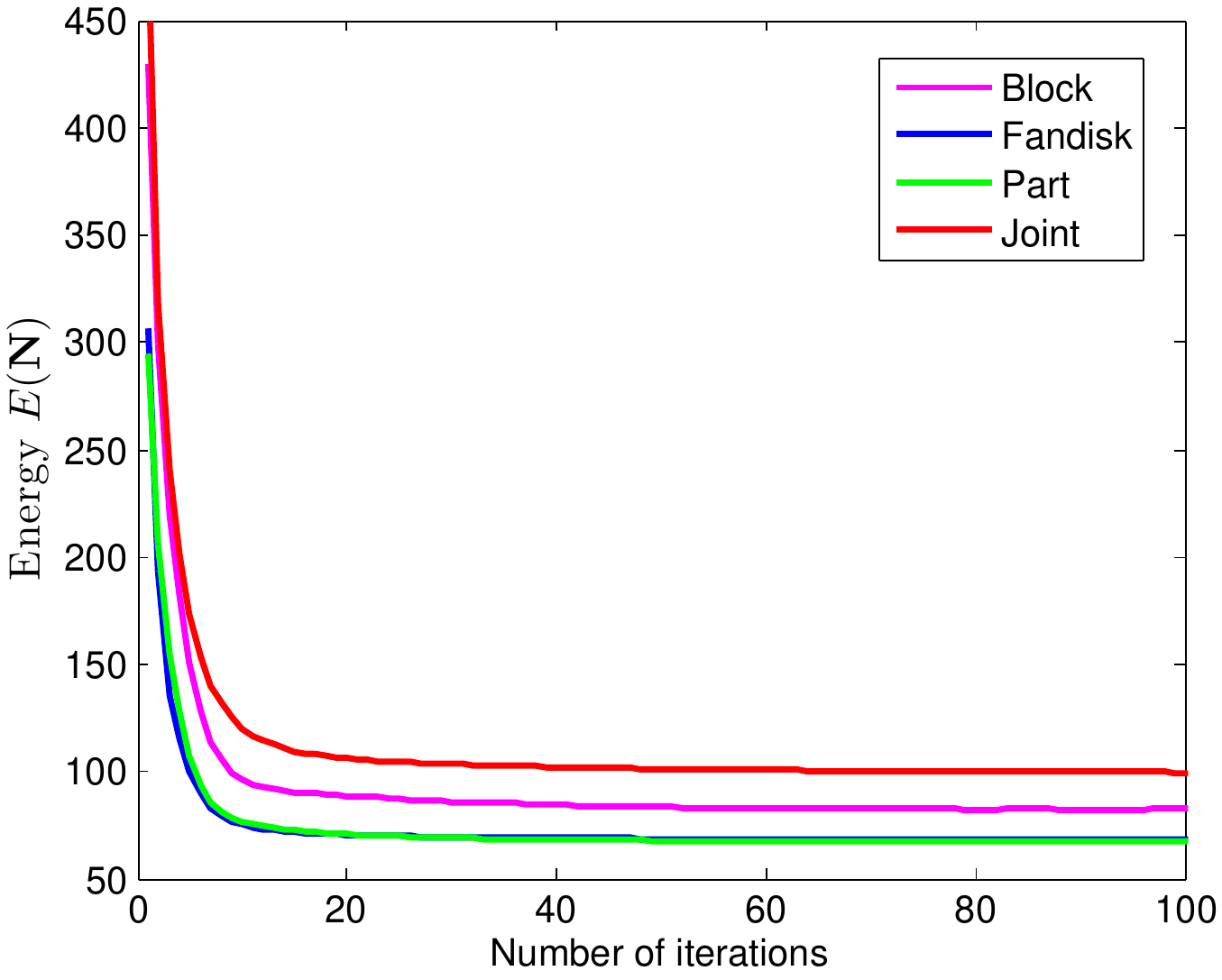}}&
  \subfloat[][]
  {\label{fig:curve-b}\includegraphics[width=0.42\textwidth]{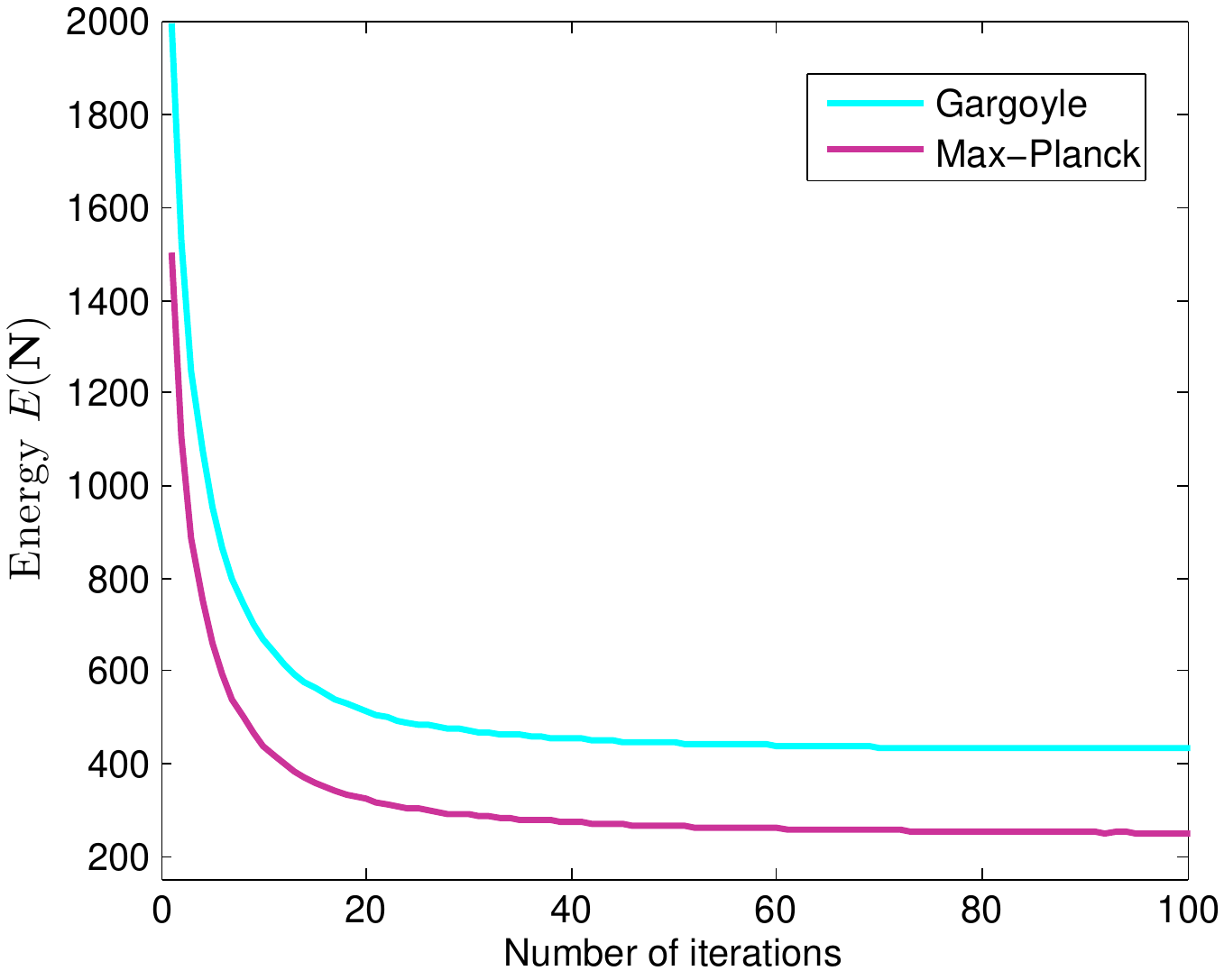}}\\
%  \subfloat[][ ]
%  {\includegraphics[width=0.45\textwidth]{bias_func}}&
   \end{tabular}
  \caption{Energy evolution via iteration numbers of (a) surfaces in \cref{fig:cadComparison} and (b) surfaces in \cref{fig:noncadComparison}.}
  \label{fig:energyCurve}
\end{figure}

Dynamic weights in the proposed normal filtering model \cref{InitHOTVmodel} play a good role in recovering sharp features of surfaces; see their effect in \cref{fig:noweightAndWeight}.
As we can see, without dynamic weights, some sharp edges are smoothed a little in the denoising procedure.
In contrast, the result with these dynamic weights is better.

\begin{figure}[tbhp]
  \centering
  \subfloat[Ground truth]{\label{fig:noweightAndWeight-a}\includegraphics[width=0.23\textwidth]{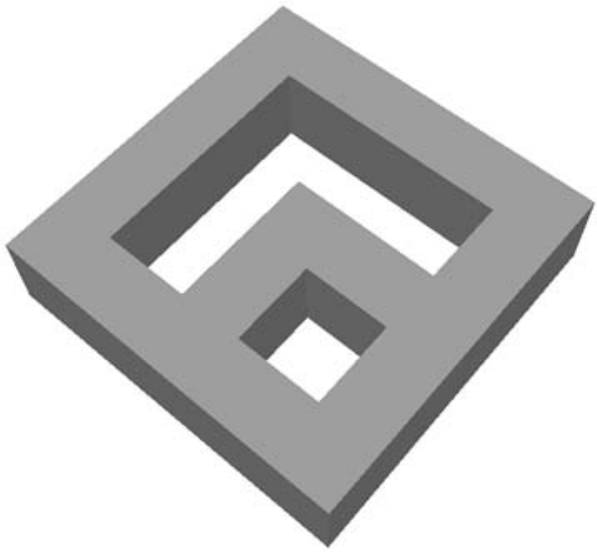}}
  \subfloat[noisy]{\label{fig:noweightAndWeight-b}\includegraphics[width=0.23\textwidth]{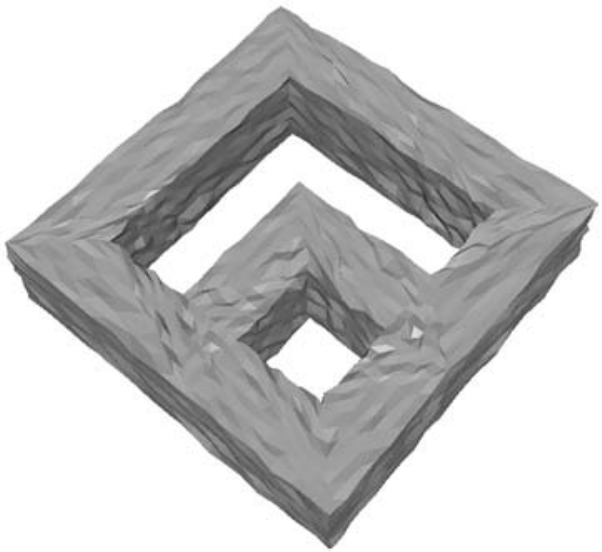}}
  \subfloat[result without dynamic weights]{\label{fig:noweightAndWeight-c}\includegraphics[width=0.23\textwidth]{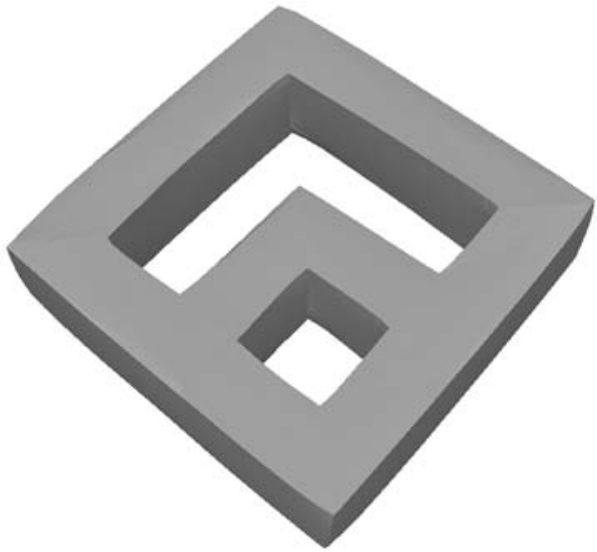}}
  \subfloat[result with dynamic weights]{\label{fig:noweightAndWeight-d}\includegraphics[width=0.23\textwidth]{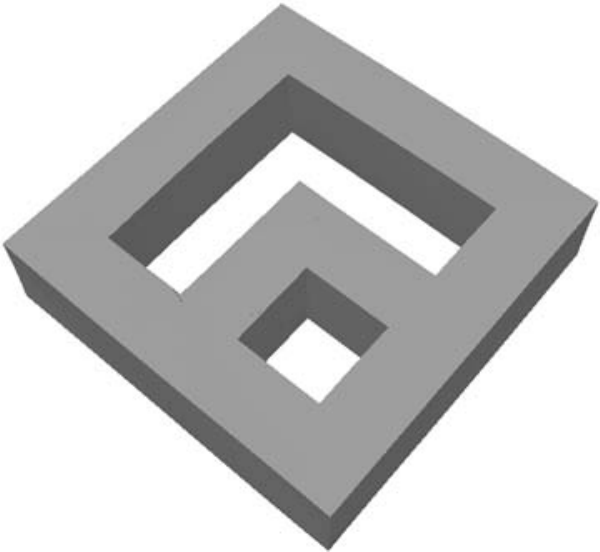}}
  \caption{Denoising results of Doubletorus (corrupted by Gaussian noise, standard deviation = 0.15 mean edge length).
  From left to right: the clean surface, input noisy surface, denoising result produced by the proposed high order normal filtering model \cref{InitHOTVmodel} without and with dynamic weights.}
  \label{fig:noweightAndWeight}
\end{figure}

%\begin{figure*}
%  \centering
%  \includegraphics[width=1.0\linewidth]{figs/noweightAndWeight}
%  \caption{Denoising results of Doubletorus (corrupted by Gaussian noise, standard deviation = 0.15 mean edge length).
%  From left to right: the clean surface, input noisy surface, denoising result produced by the proposed high order normal filtering model \cref{InitHOTVmodel} without and with dynamic weights.
%  \label{fig:noweightAndWeight}}
%\end{figure*}

\subsection{Comparison to $\ell_1$-norm Laplacian Normal Filtering Model}
In this subsection, we compare our normal filtering model \cref{InitHOTVmodel} with the $\ell_1$-norm Laplacian normal filtering model to show the advantage of our second order difference \cref{2ndOperator} over the Laplace operator \cref{LaplaceOperator} in surface denoising application.
The $\ell_1$-norm laplacian normal filtering model is given as
\begin{equation} \label{vlapmodel}
\begin{aligned}
\min \limits_{\textbf{N}\in C_\textbf{N}}\{E(\textbf{N}) = R_{\mathrm{vlap}}(\mathrm{\Delta} \textbf{N})+\frac{\alpha}{2}\left\|\textbf{N}-\textbf{N}^{in}\right\|^2_{\mathbf{V}_M}\},
\end{aligned}
\end{equation}
where
\begin{align}
C_\textbf{N}                        &=      \{\textbf{N} \in \mathbf{V}_M:\left\|\textbf{N}_\tau\right\|=1,\ \forall\tau\},\notag\\
R_{\mathrm{vlap}}(\mathrm{\Delta} \textbf{N})     &=      \sum\limits_{\tau}\mathrm{w}_{\tau}{\Big(\sum\limits_{i=1}\limits^3(\mathrm{\Delta} N_i |_{\tau})^{2}\Big)^{\frac{1}{2}}} s_{\tau}. \notag
\end{align}
The dynamic weight $\mathrm{w}_{\tau}$ on each triangle is defined as
\begin{equation*}
  %w_l = \exp(-\left\|\nabla^2 \mathbf{N} |_l\right\|^4).
  \mathrm{w}_{\tau} = \exp(- \parallel\sum\limits_{\tau_j \in D_1(\tau_i)} (\mathbf{N}_{\tau} -  \mathbf{N}_{{\tau_j}}) \parallel ^4),
\end{equation*}
which is used to enhance the sparsity of the proposed model \cref{vlapmodel}.
For fairness, our normal filtering model \cref{InitHOTVmodel} is compared with the Laplacian one \cref{vlapmodel} without and with dynamic weights respectively.
As we can see in \cref{fig:secondOrderAndLaplaceOperator}, although our normal filtering model \cref{InitHOTVmodel} and the Laplacian model \cref{vlapmodel}  both remove the staircase effect, our model \cref{InitHOTVmodel} can preserve sharp features well while the model \cref{vlapmodel} with the Laplace operator cannot.

\begin{figure}[tbhp]
  \centering
  \subfloat[]{\label{fig:secondOrderAndLaplaceOperator-a}\includegraphics[width=0.2\textwidth]{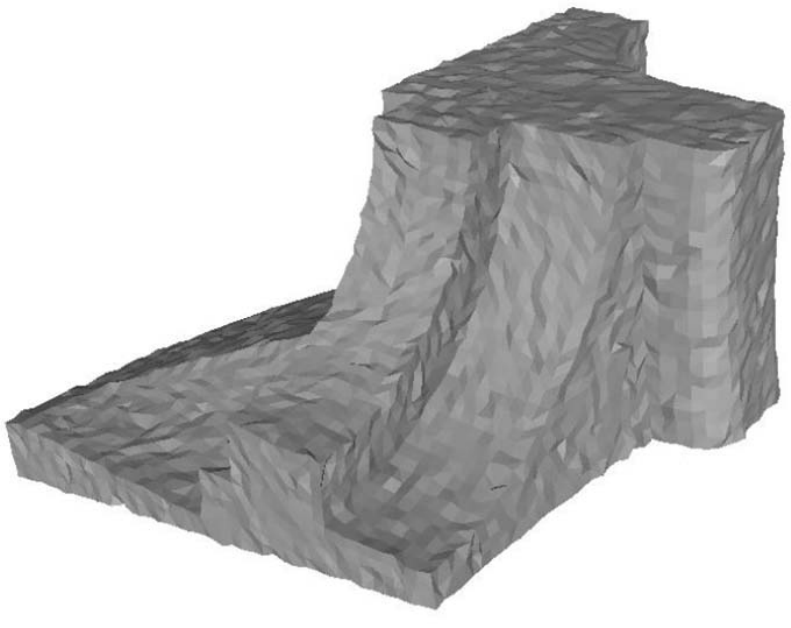}}
  \subfloat[]{\label{fig:secondOrderAndLaplaceOperator-b}\includegraphics[width=0.2\textwidth]{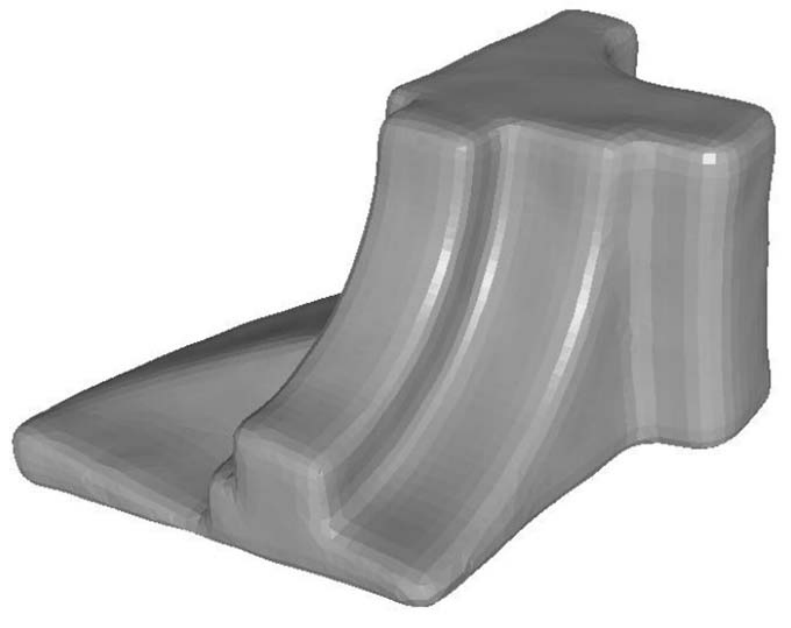}}
  \subfloat[]{\label{fig:secondOrderAndLaplaceOperator-c}\includegraphics[width=0.2\textwidth]{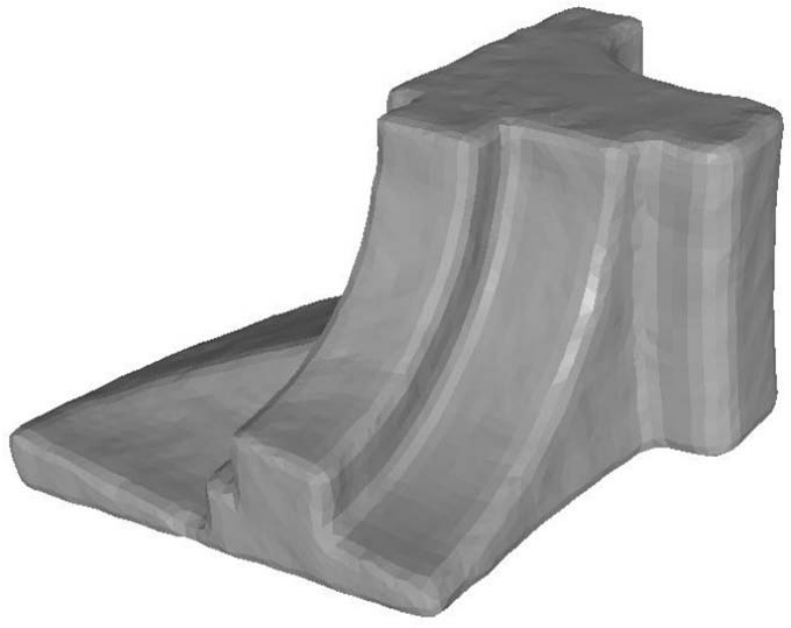}}
  \subfloat[]{\label{fig:secondOrderAndLaplaceOperator-d}\includegraphics[width=0.2\textwidth]{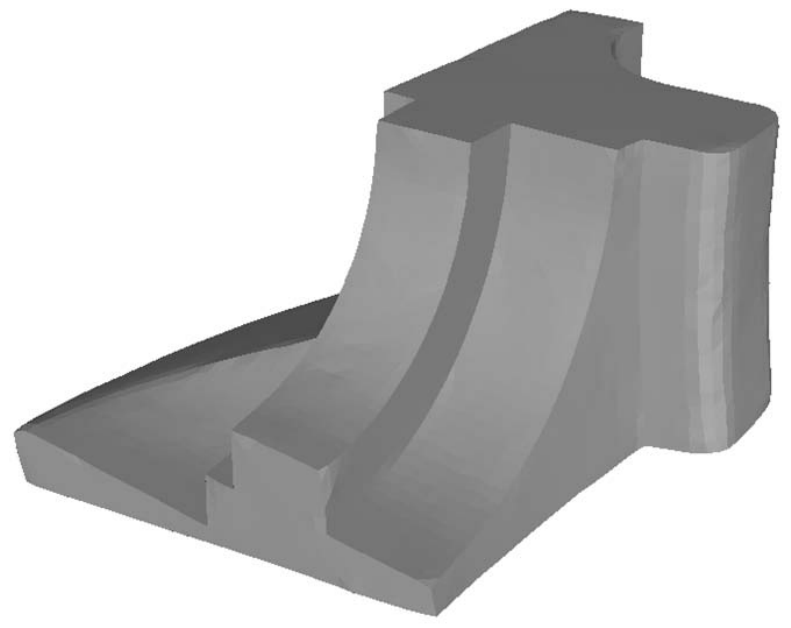}}
  \subfloat[]{\label{fig:secondOrderAndLaplaceOperator-e}\includegraphics[width=0.2\textwidth]{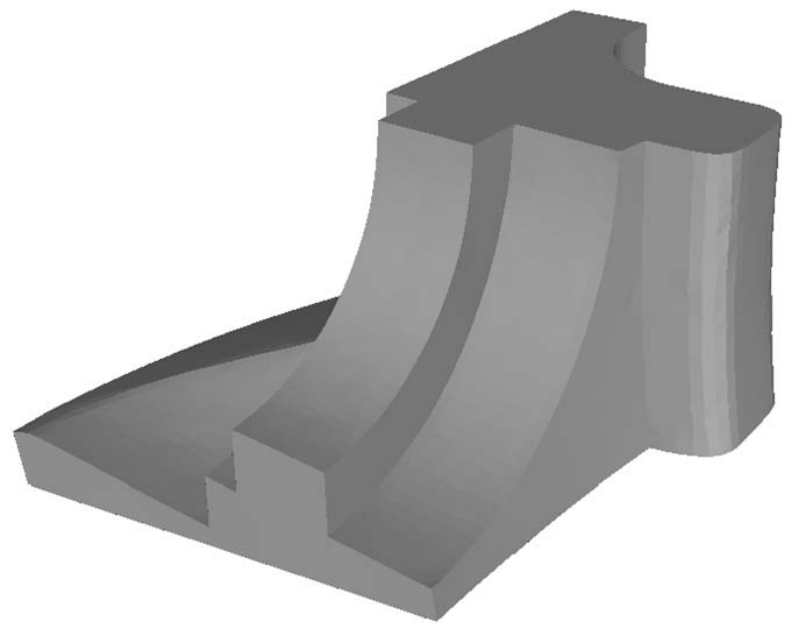}}
  \caption{Denoising results of Fandisk (corrupted by Gaussian noise, standard deviation = 0.15 mean edge length).
  (a) is the noisy surface; (b) and (c) are results produced by $\ell_1$-norm Laplacian normal filtering model \cref{vlapmodel} without and with dynamic weights respectively; (d) and (e) are results produced by the high order normal filtering model \cref{InitHOTVmodel} without and with dynamic weights respectively.}
  \label{fig:secondOrderAndLaplaceOperator}
\end{figure}

\section{Conclusion} \label{sec:conclusion}

In this paper, we propose a triangulated surface denoising meth-\\od using a newly defined discrete high order regularization.
The method applies the high order regularization to the normal vector field with a well-designed weighting function.
The variational model is solved by the augmented Lagrangian method with dynamic weights strategy.
Moreover, a new vertex updating scheme is presented to overcome the orientation ambiguities introduced by previous vertex updating methods.
We also compare our method to several denoising methods on a variety triangulated surfaces both qualitatively and quantitatively.
%demonstrating its better performance.
Conventional methods either smooth sharp features, or generate staircase artifacts.
Since our method preserves sharp features well and produces no staircase effect, it outperforms other three compared methods.
%Our method preserves sharp features well, produces no staircase effect, and outperforms other three compared methods.
Thus it can be applied to more general surfaces containing both sharp features and smoothly curved regions.

%\section*{Acknowledgments}
%We would like to acknowledge the assistance of volunteers in putting
%together this example manuscript and supplement.

\bibliographystyle{siamplain}
\bibliography{paper}
\end{document}